\pdfoutput=1
\documentclass[final,12pt]{article}
\usepackage{hyperref}
\usepackage{amssymb,amsmath,amsthm}
\usepackage[mathscr]{eucal}
\usepackage{pgf}
\usepackage{tikz}
\usetikzlibrary{arrows,backgrounds,decorations.pathmorphing,decorations.pathreplacing,positioning,fit,shapes}
\usepackage{cite}
\usepackage{enumerate}
\usepackage[conditional,light,first,bottomafter]{draftcopy}
\draftcopyName{DRAFT\space\today}{130}
\draftcopySetScale{65}
\usepackage[letterpaper,hmargin=3.7cm,vmargin=3.1cm]{geometry}
\usepackage{setspace}
\makeatletter
\renewcommand{\section}{\@startsection%
{section}%
{1}%
{0em}%
{1.7em}%
{1.2em}%
{\normalfont\large\centering\bfseries}}
\renewcommand{\@seccntformat}[1]%
{\csname the#1\endcsname.\hspace{0.5em}}

\numberwithin{equation}{section}
\renewcommand\appendix{\par
\setcounter{section}{0}%
\setcounter{subsection}{0}%
\setcounter{theorem}{0}
\setcounter{table}{0}
\setcounter{figure}{0}
\gdef\thetable{\Alph{table}}
\gdef\thefigure{\Alph{figure}}
\section*{Appendix}
\gdef\thesection{\Alph{section}}
\setcounter{section}{1}}

\newtheorem{theorem}{Theorem}[section]
\newtheorem{proposition}{Proposition}[section]
\newtheorem{lemma}{Lemma}[section]
\newtheorem{corollary}{Corollary}[section]
\theoremstyle{definition}
\newtheorem{definition}{Definition}
\newtheorem{remark}{Remark}

\newtheorem{example}{Example}

\newcommand{\norm}[1]{\left\|#1\right\|}
\newcommand{\inner}[2]{\left\langle#1,#2\right\rangle}

\newcommand{\integers}{\mathbb{Z}}
\newcommand{\tb}[1]{\widetilde{\boldsymbol{#1}}}
\newcommand{\reals}{\mathbb{R}}
\newcommand{\complex}{\mathbb{C}}
\newcommand{\nats}{\mathbb{N}}
\newcommand{\pb}[1]{\boldsymbol{#1}}

\DeclareMathOperator{\nul}{null}

\DeclareMathOperator{\spec}{spec}
\DeclareMathOperator{\Span}{span}

\DeclareMathOperator{\diag}{diag}
\begin{document}
\begin{titlepage}
\title{Spectral analysis of infinite Marchenko-Slavin matrices
\footnotetext{%
Mathematics Subject Classification(2010):
34K29,  
47B39  

70F17, 
}
\footnotetext{%
Keywords:
Inverse spectral problem;
Band symmetric matrices;
Spectral measure.
}\hspace{-5mm}
\thanks{%
Research partially supported by UNAM-DGAPA-PAPIIT IN105414
}%
}
\author{
\textbf{Sergio Palafox}
\\
\small Centro de Modelación Matemática, Vinculación y Consultoría\\[-1.6mm]
\small Universidad Tecnológica de la Mixteca\\[-1.6mm]
\small C.P. 69004, Huajuapan de León, Oaxaca\\[-1.6mm]
\small \texttt{palafox@mixteco.utm.mx}
\\[2mm]
\textbf{Luis O. Silva}
\\
\small Departamento de F\'{i}sica Matem\'{a}tica\\[-1.6mm]
\small Instituto de Investigaciones en Matem\'aticas Aplicadas y en Sistemas\\[-1.6mm]
\small Universidad Nacional Aut\'onoma de M\'exico\\[-1.6mm]
\small C.P. 04510, M\'exico D.F.\\[-1.6mm]
\small \texttt{silva@iimas.unam.mx} }
\date{}
\maketitle
\vspace{-6mm}
\begin{center}
\begin{minipage}{5in}
  \centerline{{\bf Abstract}} \bigskip This work tackles the problem of
  spectral characterization of a class of infinite matrices arising from the
  modelling of small oscillations in a system of interacting particles. The
  class of matrices under discussion corresponds to the infinite
  Marchenko-Slavin class. The spectral functions of these matrices are
  completely characterized and an algorithm is provided for the
  reconstruction of the matrix from its spectral function. The techniques
  used in this work are based on recent results for the spectral
  characterization of infinite band symmetric matrices with
  so-called degenerations.
\end{minipage}
\end{center}
\thispagestyle{empty}
\end{titlepage}

\section{Introduction}
\label{sec:intro}
It follows from Lagrangian mechanics that the mechanical characteristics of a
system of interacting particles, when the potential energy depends only on
each particle's spatial position, is given by a finite Hermitian block matrix
\cite[Chap.\,8]{MR3839313}. This matrix encodes the potential energy in the
approximation obtained when each particle is very close to the equilibrium
state. The dimension of the blocks corresponds to the degrees of freedom of
the particles. Thus, in the theory of small oscillations, any linear
arrangement of particles interacting only between immediate neighbors and
moving only along the line of interaction is mechanically equivalent to a
string of masses and springs. In this case the matrix of the system's
mechanical characteristics is a finite scalar Jacobi matrix \cite{MR2915295}.
When interactions occur not only between immediate neighbors and any particle
still has a one degree of freedom, the mechanical matrix is a scalar
Hermitian band matrix. More intricate arrangments of interacting particles
with more degrees of freedom corresponds in general to block Hermitian
matrices.

Due to the fact that Jacobi matrices are simple selfadjoint operators, the
spectral function is completely determined by the scalar measure given by its
projection onto the cyclic vector. In turn, this means that the Jacobi matrix
is reconstructed by this scalar measure. Physically, this entails that, by
observing the small oscillations of just one particle of the system, one
determines the mechanical characteristics of it and therefore the
oscillations of all the interacting particles close to the equilibrium state.
In the more general setting of arbitrary block Hermitian matrices, Marchenko
and Slavin characterized in \cite{MR3839313} a class of block Hermitian
matrices corresponding to systems of interacting particles for which the
observation of a subset of particles is sufficient for unambiguously find
their mechanical parameters. Solving this problem not only amounts to finding
the spectral multiplicity and the generating space of the corresponding
operator. The theory of extendable sets is introduced in
\cite[Chap.\,12]{MR3839313} which yields the above mentioned
characterization. A reconstruction algorithm is also provided in
\cite[Chap.\,14]{MR3839313}.

In this work, we study direct and inverse spectral problems for the
infinite dimensional generalization of the Marchenko-Salvin matrices
\cite[Chap.\,15]{MR3839313}. We restrict ourselves to the scalar case
and focus on developing a general approach suitable for a
generalization to infinite block matrices. It is worth mentioning that
the theory of extendable sets is inherently finite and, thus, for the
infinite case, one has to rely on other techniques which are rather
build up on the theory developed in \cite{MR3543793}.

The theory pertaining to the class of infinite matrices studied in
\cite{MR3543793}, denoted in this work by $\widetilde{\mathfrak M}$
and illustrated in Fig.~\ref{fig:structure-class}, is based on the
results of \cite{MR3389906} where a linear multidimensional
interpolation problem is treated (see
Section~\ref{sec:band-diag-matr}). Actually, the finite and infinite
classes of matrices given in \cite{MR3711273} and \cite{MR3543793},
respectively, were tailor-made to fit the structure of the solutions
to the interpolation problem. Indeed, the vector polynomials arising
from the so-called ``degenerations'' of the matrix (see
Fig.~\ref{fig:structure-class}) generate the solutions to the
corresponding linear multidimensional interpolation problem
\cite[Sec.\,4]{MR3389906}. Noteworthly, the degenerations of the
matrix are interpreted as \emph{inner boundary conditions}
\cite[Sec.\,2]{MR3711273} of the associated difference equation and
are relevant in the theory of such equations.

The infinite dimensional Marchenko-Slavin class, introduced in
Definition~\ref{def:class}, illustrated in Fig.~\ref{fig:gral-example}
and denoted by $\mathfrak M$, contains properly the class
$\widetilde{\mathfrak M}$. However, it is proven in
Theorem~\ref{thm:main}, that every matrix in $\mathfrak M$ is unitary
equivalent to a matrix in $\widetilde{\mathfrak M}$. This equivalence
between $\mathfrak M$ and $\widetilde{\mathfrak M}$ has implications
in both the linear multidimensional interpolation problem and the
modelling of systems of interacting particles. Indeed, a system of the
more general type modelled by a matrix in $\mathfrak M$ can always be
reduced to a system whose mechanical characteristics are given by a
simpler matrix in $\widetilde{\mathfrak M}$.

Proving the equivalence of the clases $\mathfrak M$ and
$\widetilde{\mathfrak M}$ requires firstly a complete characterization
of the spectral functions of a given matrix $M\in\mathfrak M$ (see
Definition~\ref{def:spectral-measure-gen-case}) and, secondly, the
application of the reconstruction algorithm developed in
\cite{MR3543793}. For finding the spectral functions of the infinite
matrix, one needs to characterize the spectral functions of finite
submatrices of a matrix $M$ (see Section~\ref{sec:preliminaries}) and
recur to a limit process. This involves solving a problem related to
the finite matrix moment problem (see
Theorem~\ref{thm:moments-finite}).

\section{The class $\mathfrak M$ of semi-infinite matrices}
\label{sec:preliminaries}
Let $l_{2}(\nats)$ be the Hilbert space of square summable sequences
with entries in $\complex$ and $\{\delta_k\}_{k=1}^{\infty}$ be the
canonical orthonormal basis of it. Along the text, a sequence
$\{u_{k}\}_{k\in\nats}$ is identified with
$u=\sum_{k=1}^{\infty}u_{k}\delta_{k}$.  Also, we consider square
summable sequences $\{u_{k}\}_{k\in G}$, where $G\subset\nats$, and
denote the corresponding space by $l_{2}(G)$.
\begin{definition}
  \label{def:projectors}
  For $F\subset G$ and
  any $u\in l_{2}(G)$, define:
  \begin{enumerate}[(i)]
  \item $\Pi_{G\to F}u:=\{u_{k}\}_{k\in G}$, where
    $u_{k}=
    \begin{cases}
      u_{k}=u_{k}, & k\in F\\
      u_{k}=0, & k\in G\setminus F
    \end{cases}
    $
  \item $\widetilde\Pi_{G\to F}u:=\{u_{k}\}_{k\in F}$
\end{enumerate}
\end{definition}
Note that $\Pi_{G\to F}$ is an orthonormal projector while
$\widetilde\Pi_{G\to F}$ is not.

Let us agree on the following terminology. The nonzero entries of a
matrix to the right of which there are only zero entries on the same
row are called \emph{row-edge entries}. Similarly, the nonzero entries
of a matrix for which there are only zero entries on the same column
are called \emph{column-edge entries}.

The following matrices form the class of infinite Marchenko-Slavin
matrices (see in \cite{MR3839313} the finite analogue of this class).
\begin{definition}
  \label{def:class}
  The matrix $M=\{m_{jk}\}_{j,k\in\nats}$ is in the class
  $\mathfrak M$ when $m_{jk}=\overline{m}_{kj}$ for any $j,k\in\nats$
  and the following conditions are met.
  \begin{enumerate}[(1)]
  \item There is $l\in\nats$ such that there is only one
    \emph{row-edge entry} on each of the columns
    $l+1,l+2,\dots$. Denote either by $n_{M}$ or $n$ the minimal $l$
    with this property.\label{cond-q}
  \item There are numbers $j,k\in\nats$ with $k>j$ such that
    \begin{equation*}
    m_{jk},\,m_{j+1k+1},\,m_{j+2k+2},\dots
  \end{equation*}
  are simultaneously row-edge and column-edge entries. Denote by
  $j_{0},k_{0}$ the smallest $j,k$ having this
  property.\label{cond-tail}
\end{enumerate}
\end{definition}
The condition (2) of Definition~\ref{def:class} gives the ``tail'' of
the matrix. An illustration of a matrix in $\mathfrak M$ is given in
Fig.~\ref{fig:gral-example}.

\begin{figure}[h]
    \centering
\begin{tikzpicture}[scale=.18]\footnotesize
 \pgfmathsetmacro{\xone}{0}
 \pgfmathsetmacro{\xtwo}{ 22.6}
 \pgfmathsetmacro{\yone}{7.6}
 \pgfmathsetmacro{\ytwo}{30}
  \draw[step=1cm,gray,opacity=0.5,very thin] (\xone,\yone) grid (\xtwo,\ytwo);
\draw(22.3,7.3)node[scale=.8]{$\ddots$};
\draw(21.3,7.3)node[scale=.8]{$\ddots$};
\draw(20.3,7.3)node[scale=.8]{$\ddots$};
\draw(23.3,8.5)node[scale=.8]{$\ddots$};
\draw(23.3,9.5)node[scale=.8]{$\ddots$};
\draw(23.3,10.5)node[scale=.8]{$\ddots$};
\draw(23.3,7.6)node[scale=.8]{$\ddots$};
\draw(36.5,17.5)node[text width=2.5cm,scale=1]{A zero entry};
\draw(36.5,15.5)node[text width=2.5cm,scale=1]{Any entry}; 
\draw(36.5,13.5)node[text width=2.5cm,scale=1]{A nonzero entry};
\draw[step=1cm,gray,opacity=0.5,very thin] (27,17)
  grid (28,18); 
\draw[step=1cm,gray,opacity=0.5,very thin] (27,15) grid
  (28,16); 
\filldraw[thin,gray,opacity=.4] (27,15) rectangle (28,16);
\filldraw[thin,gray,opacity=1] (27,13) rectangle (28,14);

\begin{scope}
\foreach \x in {9,10,11,12,13,14,15,16,17,18,19,20,21,22,23,24,25,26,27}
\draw(23.3,2.5+\x)node[scale=.8]{$\dots$}
;
\end{scope}

\begin{scope}
\foreach \x in {0,1,2,3,4,5,6,7,8,9,10,11,12,13,14,15,16,17,18}
\draw(\x+.5,7.5)node[scale=.8]{$\vdots$}
;
\end{scope}
\begin{scope}
  \filldraw[thin,gray,opacity=.9] (1, 19) rectangle (2,20);
  \filldraw[thin,gray,opacity=.9] (2, 17) rectangle (3,18);
  \filldraw[thin,gray,opacity=.9] (4, 20) rectangle (5,21);
  \filldraw[thin,gray,opacity=.9] (5, 14) rectangle (6,15);
  \filldraw[thin,gray,opacity=.9] (6, 18) rectangle (7,19);
  \filldraw[thin,gray,opacity=.9] (7, 16) rectangle (8,17);
  \filldraw[thin,gray,opacity=.9] (9, 15) rectangle (10,16);
  \filldraw[thin,gray,opacity=.9] (11, 13) rectangle (12,14);
  \filldraw[thin,gray,opacity=.9] (10, 28) rectangle (11,29);
  \filldraw[thin,gray,opacity=.9] (12, 27) rectangle (13,28);
  \filldraw[thin,gray,opacity=.9] (9, 25) rectangle (10,26);
  \filldraw[thin,gray,opacity=.9] (11, 23) rectangle (12,24);
  \filldraw[thin,gray,opacity=.9] (13, 22) rectangle (14,23);
  \filldraw[thin,gray,opacity=.9] (14, 20) rectangle (15,21);
  \filldraw[thin,gray,opacity=.9] (15, 24) rectangle (16,25);
  \filldraw[thin,gray,opacity=.9] (16, 18) rectangle (17,19);
  \filldraw[thin,gray,opacity=.4] (0, 22) rectangle (1,26);
  \filldraw[thin,gray,opacity=.4] (1, 20) rectangle (2,25);
  \filldraw[thin,gray,opacity=.4] (2, 18) rectangle (3,24);
  \filldraw[thin,gray,opacity=.4] (3, 12) rectangle (4,23);
  \filldraw[thin,gray,opacity=.4] (4, 21) rectangle (5,22);
  \filldraw[thin,gray,opacity=.4] (5, 15) rectangle (6,21);
  \filldraw[thin,gray,opacity=.4] (6, 19) rectangle (7,20);
  \filldraw[thin,gray,opacity=.4] (7, 17) rectangle (8,19);
  \filldraw[thin,gray,opacity=.4] (8, 12) rectangle (9,18);
  \filldraw[thin,gray,opacity=.4] (9, 16) rectangle (10,17);
  \filldraw[thin,gray,opacity=.4] (10, 12) rectangle (11,16);
  \filldraw[thin,gray,opacity=.4] (11, 14) rectangle (12,15);
  \filldraw[thin,gray,opacity=.4] (12, 12) rectangle (13,14);
  \filldraw[thin,gray,opacity=.4] (4, 30) rectangle (8,29);
  \filldraw[thin,gray,opacity=.4] (5, 29) rectangle (10,28);
  \filldraw[thin,gray,opacity=.4] (6, 28) rectangle (12,27);
  \filldraw[thin,gray,opacity=.4] (7, 27) rectangle (18,26);
  \filldraw[thin,gray,opacity=.4] (8, 26) rectangle (9,25);
  \filldraw[thin,gray,opacity=.4] (9, 25) rectangle (15,24);
  \filldraw[thin,gray,opacity=.4] (10, 24) rectangle (11,23);
  \filldraw[thin,gray,opacity=.4] (11, 23) rectangle (13,22);
  \filldraw[thin,gray,opacity=.4] (12, 22) rectangle (18,21);
  \filldraw[thin,gray,opacity=.4] (13, 21) rectangle (15,20);
  \filldraw[thin,gray,opacity=.4] (14, 20) rectangle (18,19);
  \filldraw[thin,gray,opacity=.4] (15, 19) rectangle (16,18);
  \filldraw[thin,gray,opacity=.4] (16, 18) rectangle (18,17);
\end{scope}
\begin{scope}
\foreach \x in {0}
{
  \filldraw[thin,gray,opacity=.9] (0+\x, 21-\x)
    rectangle (1+\x,22-\x)
 ;
   \filldraw[thin,gray,opacity=.9] (8+\x, 30-\x)
     rectangle (9+\x,29-\x);}
 \end{scope}
\begin{scope}
\foreach \x in {13}
{
  \filldraw[thin,gray,opacity=.9] (0+\x, 25-\x)
    rectangle (1+\x,26-\x)
 ;
   \filldraw[thin,gray,opacity=.9] (4+\x, 30-\x)
     rectangle (5+\x,29-\x);}
\end{scope}
\begin{scope}
\foreach \x in {0,1,2,3,4,5,6,7,8,9,10,11,12,13,14}
{
  \filldraw[thin,gray,opacity=.4] (0+\x, 26-\x)
    rectangle (1+\x,27-\x)
 ;
   \filldraw[thin,gray,opacity=.4] (3+\x, 30-\x)
     rectangle (4+\x,29-\x);}
\end{scope}
\begin{scope}
\foreach \x in {15,16,17,18}
{
  \filldraw[thin,gray,opacity=.9] (0+\x, 26-\x)
    rectangle (1+\x,27-\x)
 ;
   \filldraw[thin,gray,opacity=.9] (3+\x, 30-\x)
     rectangle (4+\x,29-\x);}
\end{scope}
\begin{scope}
\foreach \x in {0,1,2,3,4,5,6,7,8,9,10,11,12,13,14,15,16,17,18,19}
{
  \filldraw[thin,gray,opacity=.4] (0+\x, 27-\x)
    rectangle (1+\x,28-\x)
 ;
   \filldraw[thin,gray,opacity=.4] (2+\x, 30-\x)
     rectangle (3+\x,29-\x);}
\end{scope}
\begin{scope}
\foreach \x in {0,1,2,3,4,5,6,7,8,9,10,11,12,13,14,15,16,17,18,19,20}
{
  \filldraw[thin,gray,opacity=.4] (0+\x, 28-\x)
    rectangle (1+\x,29-\x)
 ;
   \filldraw[thin,gray,opacity=.4] (1+\x, 30-\x)
     rectangle (2+\x,29-\x);}
\end{scope}
\begin{scope}
  \foreach \x in
  {0,1,2,3,4,5,6,7,8,9,10,11,12,13,14,15,16,17,18,19,20,21}
  { \filldraw[thin,gray,opacity=.25] (0+\x, 29-\x) rectangle
    (1+\x,30-\x) ; \filldraw[thin,gray,opacity=.2] (0+\x, 29-\x)
    rectangle (1+\x,30-\x);}
\end{scope}
\end{tikzpicture}
\caption{An example of an element of $\mathfrak{M}$}
\label{fig:gral-example}
\end{figure}
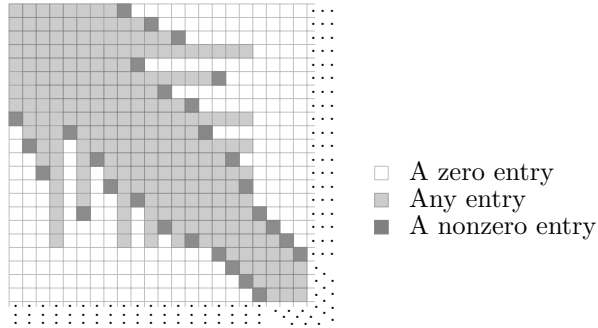

According to \cite[Sec.\,47 Thm.\,4]{MR1255973}, given an orthonormal
basis in $l_{2}(\nats)$, one can construct uniquely from the matrix
$M$ a closed symmetric operator, which is denoted by $\mathcal{M}$, in
such a way that $M$ is its matrix representation with respect to the
given orthonormal basis (see \cite[Sec.\,47]{MR1255973}).

The analysis of the spectral function of operator $\mathcal M$ is
carried out by means of the auxiliary operator
$\mathcal M_{N}:=\widetilde\Pi_{\nats\to G_{N}}\mathcal M
\upharpoonright_{l_{2}(G_{N})}$, where $G_{s}:=\{1,\dots,s\}$ and
$N>n$. Note that $\mathcal M_{N}$ can be identified with the operator
whose matrix representation is the finite dimensional submatrix
corresponding to the $N\times N$ upper-left corner of $M$ which is
denoted by $M_{N}$. Note that, since $N>n$, any matrix $M_N$
determines a set $K\subset G_{N}$ such that for each $k\in K$ the
$k$-th row has no row-edge entry and the cardinality of $K$ is always
equal to $n$ (see Definition~\ref{def:class} (1)).

\section{Spectral functions of submatrices}
\label{sec:spectr-meas-subm}

For a matrix $M\in\mathfrak M$, consider the finite submatrix
$M_{N}$. To study the spectral properties of the operator
$\mathcal{M}_{N}$, one recurs to the equation
\begin{equation}
  \label{eq:formal-difference}
  M_{N}\psi=z \psi\,,
\end{equation}
for $\psi\in l_{2}(G_{N})$. To solve this equation, one starts by
finding the solution to
\begin{equation}
  \label{eq:formal-difference_p}
  \Pi_{G_{N}\to K^{\perp}}(M_{N}-zI) \psi=0
\end{equation}
where $K^{\perp}$ is the complement in $G_{N}$ of the set $K$ given at
the end of the previous section.

The solution to \eqref{eq:formal-difference_p} is constructed
recursively as follows: one gives the first $n$ elements of the
sequence $\psi$ and finds $\psi_{n+1}$ from the first row equation of
\eqref{eq:formal-difference_p} (see \eqref{eq:example-system1} in
Example~\ref{ex:simple-difference-equation} below). It is possible to
solve this equation due to the fact that the row-edge entry, which is
on the $n+1$-th column, is not zero. Having found $\psi_{n+1}$, one
finds the subsequent elements of the solution using the row equations
where there are row-edge entries (that is row equations corresponding
to $K^{\perp}$).
\begin{example}
  \label{ex:simple-difference-equation}
  If one has the matrix $M_{7}$ corresponding to
  Fig.~\ref{fig:simple-example}, then the difference equations
  corresponding to the rows are
  \begin{align}
    \label{eq:example-system1}
    m_{11}\psi_1+m_{12}\psi_2+m_{13}\psi_3+\pb{m}_{14}\psi_4&=z \psi_1\\
    m_{21}\psi_1+m_{22}\psi_2+m_{23}\psi_3+m_{24}\psi_4+m_{25}\psi_5+\pb{m}_{26}\psi_6&=z\psi_2\nonumber
    \\\label{eq:example-system3}
    m_{31}\psi_1+m_{32}\psi_2+m_{33}\psi_3+m_{34}\psi_4+m_{35}\psi_5+m_{36}\psi_6&=z \psi_3\\
    m_{41}\psi_1+m_{42}\psi_2+m_{43}\psi_3+m_{44}\psi_4+\pb{m}_{45}\psi_5&=z \psi_4\nonumber\\\label{eq:example-system5}
    m_{52}\psi_2+m_{53}\psi_3+m_{54}\psi_4+m_{55}\psi_5+m_{56}\psi_6&=z
                                                                      \psi_5\\
    m_{62}\psi_2+m_{63}\psi_3+m_{65}\psi_5+m_{66}\psi_6+\pb{m}_{67}\psi_7&=z \psi_6\nonumber\\
    m_{76}\psi_6+m_{77}\psi _7&=z \psi_7\label{eq:example-system7}
  \end{align}
  where the row-edge entries have been donoted in bold typeface. Note
  that in \eqref{eq:example-system3}, \eqref{eq:example-system5} and
  \eqref{eq:example-system7} there are no row-edge entries. In this
  case, the system \eqref{eq:formal-difference}, without equations
  \eqref{eq:example-system3}, \eqref{eq:example-system5},
  \eqref{eq:example-system7} (\emph{i.\,e.}, the system corresponding
  to \eqref{eq:formal-difference_p}) and with boundary conditions
  given by the numbers $\psi _{1},\psi _{2},\psi _{3}$, has the
  following solution:
\begin{equation}
\begin{pmatrix}
  \psi _1\\
  \psi _2\\
  \psi _3\\
  \psi_4\\
  \psi_5\\
  \psi_6\\
  \psi_7\\
\end{pmatrix}=
\begin{pmatrix}
  \psi_1\\
  \psi_2\\
  \psi_3\\
  \pb{m}_{14}^{-1}(z \psi_1-m_{11}\psi_1-m_{12}\psi_2-m_{13}\psi_3)\\
  \pb{m}_{45}^{-1}(z \psi_4-m_{41}\psi_1-m_{42}\psi_2-m_{43}\psi_3-m_{44}\psi_4)\\
  {\pb m}_{26}^{-1}(z \psi_2-m_{21}\psi_1-m_{22}\psi_2-m_{23}\psi_3-m_{24}\psi_4-m_{25}\psi_5)\\
  {\pb m}_{67}^{-1}(z \psi_6-m_{62}\psi_2-m_{63}\psi_3-m_{65}\psi_5-m_{66}\psi_6)
\end{pmatrix}\,.
\label{eq:solution-U}
\end{equation}
\end{example}
\begin{figure}[h]
    \centering
\begin{tikzpicture}[scale=.4]\footnotesize
 \pgfmathsetmacro{\xone}{0}
 \pgfmathsetmacro{\xtwo}{7}
 \pgfmathsetmacro{\yone}{1}
 \pgfmathsetmacro{\ytwo}{8}
  \draw[step=1cm,gray,opacity=0.5,very thin] (\xone,\yone) grid (\xtwo,\ytwo);
\begin{scope}
  \filldraw[thin,gray,opacity=.9] (3, 7) rectangle (4,8);
  \filldraw[thin,gray,opacity=.9] (4, 4) rectangle (5,5);
  \filldraw[thin,gray,opacity=.9] (5, 6) rectangle (6,7);
  \filldraw[thin,gray,opacity=.9] (6, 2) rectangle (7,3);
  \filldraw[thin,gray,opacity=.9] (0, 4) rectangle (1,5);
 \filldraw[thin,gray,opacity=.9] (1, 2) rectangle (2,3);
 \filldraw[thin,gray,opacity=.9] (3, 3) rectangle (4,4);
 \filldraw[thin,gray,opacity=.9] (5, 2) rectangle (6,1);
\end{scope}
\begin{scope}
\foreach \x in {0,1}
{
  \filldraw[thin,gray,opacity=.4] (1+\x, 3-\x)
    rectangle (2+\x,4-\x)
 ;
   \filldraw[thin,gray,opacity=.4] (4+\x, 7-\x)
     rectangle (5+\x,6-\x);}
\end{scope}
\begin{scope}
\foreach \x in {0,1,2}
{
  \filldraw[thin,gray,opacity=.4] (0+\x, 5-\x)
    rectangle (1+\x,6-\x)
 ;
   \filldraw[thin,gray,opacity=.4] (2+\x, 8-\x)
     rectangle (3+\x,7-\x);}
\end{scope}
\begin{scope}
\foreach \x in {0,1,2,4}
{
  \filldraw[thin,gray,opacity=.4] (0+\x, 6-\x)
    rectangle (1+\x,7-\x)
 ;
   \filldraw[thin,gray,opacity=.4] (1+\x, 8-\x)
     rectangle (2+\x,7-\x);}
\end{scope}
\begin{scope}
  \foreach \x in
  {0,1,2,3,4,5,6}
  { \filldraw[thin,gray,opacity=.25] (0+\x, 7-\x) rectangle
    (1+\x,8-\x) ; \filldraw[thin,gray,opacity=.2] (0+\x, 7-\x)
    rectangle (1+\x,8-\x);}
\end{scope}
\end{tikzpicture}
\caption{A submatrix of $\mathfrak{M}$ with $N=7$ and $n=3$.}
\label{fig:simple-example}
\end{figure}
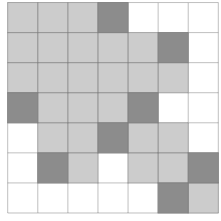

As was already mentioned, the first $n$ elements of the sequence
$\psi$ play the role of boundary conditions for the difference
equation. For assigning these conditions, we introduce an $n\times n$
upper triangular, invertible matrix $\mathscr{T}=\{t_{ij}\}_{i,j=1}^n$
which will be called \emph{boundary matrix}. Let
$\psi^j_{\mathscr{T}}(z)$ ($j\in G_n$) be an $N$-dimensional vector
solution to \eqref{eq:formal-difference_p} such that
$\inner{\delta_i}{\psi_{\mathscr T}^j(z)}=t_{ji}$ for $i\in
G_n$. Thus,
\begin{equation}
  \label{eq:psi_tau}
  \widetilde\Pi_{G_{N}\to G_{n}}\Psi_{\mathscr{T}}(z)=\mathscr{T}^*\,,
\end{equation}
where $\Psi_{\mathscr{T}}(z)$ is the $N\times n$ matrix whose columns
are the solutions
$\psi^1_{\mathscr{T}}(z),\dots,\psi^n_{\mathscr{T}}(z)$. We may use
the notation:
\begin{equation}
  \label{eq:notation-colums-composition}
  \Psi_{\mathscr{T}}(z)=
\begin{pmatrix}
  \psi^1_{\mathscr{T}}(z)&\psi^2_{\mathscr{T}}(z)&\dots&\psi^{n}_{\mathscr{T}}(z)
\end{pmatrix}\,.
\end{equation}
\begin{lemma}
  \label{lem:ind_linear_solution}
  For any $z\in\complex$, the $N$-dimensional vector $\eta(z)$ is a
  solution to \eqref{eq:formal-difference_p} if and only if
  $\eta(z)=\Psi_{\mathscr T}(z) C$, where $C$ is an $n$-dimensional
  vector.  If $C_{1}$ and $C_{2}$ are two linearly independent vectors
  in $l_{2}(G_{n})$, then $\Psi_{\mathscr T}(z) C_{1}$ and
  $\Psi_{\mathscr T}(z) C_{2}$ are linearly independent solutions to
  \eqref{eq:formal-difference_p}.
\end{lemma}
\begin{proof}
  The fact that $\Psi_{\mathscr T}(z) C$ is a solution to
  \eqref{eq:formal-difference_p} is straightforward. Let $\eta$ be a
  solution to \eqref{eq:formal-difference_p} for a fixed $z$. Since
  $\mathscr T$ is a triangular invertible matrix, there is a vector
  $C_{\eta}$ such that $\Psi_{\mathscr T}(z)C_{\eta}$ has the same
  first $n$ entries as $\eta$. Therefore
  $\eta-\Psi_{\mathscr T}(z)C_{\eta}$ is the trivial solution (see the
  paragraph above Example~\ref{ex:simple-difference-equation}). For
  proving the second part, one uses that $\nul\Psi(z)=0$ which again
  follows from the invertibility of $\mathscr T$.
\end{proof}
  \begin{proposition}
    \label{prop:spectrum}
  Let $M_N$ be a submatrix of $M\in\mathfrak M$. Define
     \begin{equation}
       \label{eq:Q_matrix}
       \Theta_{\mathscr T}(z):=\widetilde\Pi_{G_{N}\to K}(M_{N}-zI) \Psi_{\mathscr T}(z)\,.
     \end{equation}
     The spectrum of $M_N$ is given by the zeros of the polynomial
     $\det\Theta_{\mathscr T}(z)$. Moreover, if
     $\lambda\in\spec(M_{N})$, then there is a nonzero vector
     $C^{\lambda}$ such that the corresponding normalized eigenvector
     $\varphi^{\lambda}$ is given by
       \begin{equation}
    \label{eq:relation_phi_psi}
    \varphi^{\lambda}=\Psi_{\mathscr{T}}(\lambda)C^{\lambda}\,.
  \end{equation}
   \end{proposition}
   \begin{proof}
     If $\lambda$ is an eigenvalue and $\varphi^{\lambda}$ is the
     corresponding eigenvector of $M_{N}$, then
     \eqref{eq:formal-difference} and \eqref{eq:formal-difference_p}
     are satisfied for $z=\lambda$ and $\psi=\varphi^{\lambda}$. Due
     to Lemma~\ref{lem:ind_linear_solution} there is a nonzero vector
     $C^{\lambda}$ such that
     $\varphi^{\lambda}=\Psi_{\mathscr{T}}(\lambda)C^{\lambda}$. Also,
\begin{align}
  0&=\widetilde{\Pi}_{G_{N}\to K}(M_{N}-\lambda I) \Psi_{\mathscr
     T}(\lambda)C^{\lambda}\nonumber\\
   &= \Theta_{\mathscr T}(\lambda)C^{\lambda}\label{eq:formal-difference_p_dif}
   \end{align}
   Hence, the homogeneous linear system
   \eqref{eq:formal-difference_p_dif} has a nontrivial solution if and
   only if $\det \Theta_{\mathscr T}(\lambda)=0$. The normalization of
   the vector $\varphi^{\lambda}$ can clearly be achieved by modifying
   $C^{\lambda}$ appropriately.
\end{proof}

\begin{corollary}
 \label{cor:eigenspaces}
 The multiplicity of the eigenvalues of $M_{N}$ is no greater than
 $n$. If $ \{\lambda_k\}_{k=1}^N $ is the spectrum of the matrix
 enumerated so that the multiplicity of eigenvalues is taken into
 account, then there is a finite sequence of vectors,
 $\{\varphi^{\lambda_{k}}\}_{k=1}^{N}$, being an orthonormal basis of
 $l_{2}(G_{N})$.
\end{corollary}
\begin{proof}
  It follows from \eqref{eq:relation_phi_psi} that the dimension of
  the eigenspaces is at most $n$ since $C^{\lambda}$ is
  $n$-dimensional. The second assertion of the corollary is a
  consequence of $M_{N}$ being selfadjoint.
\end{proof}

Henceforth, the spectrum of $M_{N}$ is written as
$ \{\lambda_k\}_{k=1}^N$, where the eigenvalues are enumerated so that
the multiplicity of them is taken into account, and the corresponding
normalized eigenvectors are denoted by
$\varphi^{\lambda_{1}},\dots,\varphi^{\lambda_{N}}$. Also,
$C_{\mathscr T}^{\lambda_{k}}$ is the nonzero $n$-dimensional vector
determined by $\varphi^{\lambda_{k}}$ through
\eqref{eq:relation_phi_psi}.

Using the notation introduced in
\eqref{eq:notation-colums-composition}, one defines
\begin{equation}
\label{eq:spectral_data}
\Lambda:=\diag\{\lambda_1,\dots,\lambda_N\}\quad\text{and}\quad\Phi:=
\begin{pmatrix}
 \varphi^{\lambda_{1}}&\cdots&\varphi^{\lambda_{N}}
\end{pmatrix}\,.
\end{equation}
Thus
\begin{equation}
\label{eq:phi-normalization}
  \Phi^*M_{N}\Phi=\Lambda\,,\quad \text{and} \quad\Phi\Phi^*=\Phi^*\Phi=I\,.
\end{equation}

Let us also define
\begin{equation}
  \label{eq:phi0}
  \Phi_0:=\widetilde\Pi_{G_N\to G_n}\Phi\,.
\end{equation}
\begin{lemma}
  \label{lem:restricted-identity}
  Let
  $I_n:=\widetilde\Pi_{G_N\to
    G_n}I\upharpoonright_{l_{2}(G_{N})}$. The following decomposition
  takes place
  \begin{equation*}
    I_{n}=\mathscr{T}^*\sum_{k=1}^NC_{\mathscr{T}}^{\lambda_{k}}
    (C_{\mathscr{T}}^{\lambda_{k}})^*\mathscr{T}\,.
\end{equation*}
\end{lemma}
\begin{proof}
  One verifies, on the basis of the unitary properties of $\Phi$, that
  \begin{equation*}
    \Phi_0\Phi_0^*=I_{n}\,.
  \end{equation*}
  Note, however that $\Phi_0^*\Phi_0\neq I_n$. If
  $\varphi_0^{\lambda_{k}}:=\widetilde\Pi_{G_N\to
    G_n}\varphi^{\lambda_{k}}=\Phi_0\delta_k$, then
\begin{equation}
   \label{eq:partial-unitarity}
   I_{n}=\Phi_0\Phi_0^*=\Phi_0\left(\sum_{k=1}^N\delta_k\delta_k^*\right)\Phi_0^*=\sum_{k=1}^N\Phi_0\delta_k\delta_k^*\Phi_0^*=\sum_{k=1}^N\varphi_0^{\lambda_{k}}(\varphi_0^{\lambda_{k}})^{*}\,.
\end{equation}
Now, it follows from \eqref{eq:psi_tau} and
\eqref{eq:relation_phi_psi} that
\begin{equation}
  \label{eq:other-varphi-0}
  \mathscr{T}^*C_{\mathscr{T}}^{\lambda_{k}}=\varphi_0^{\lambda_{k}}\,.
\end{equation}
Thus, combining \eqref{eq:partial-unitarity} and
\eqref{eq:other-varphi-0}, one has
\begin{equation*}
  I_{n}=\sum_{k=1}^N\varphi_0^{\lambda_{k}}(\varphi_0^{\lambda_{k}})^*=
  \sum_{k=1}^N\mathscr{T}^*C_{\mathscr{T}}^{\lambda_{k}}
  (C_{\mathscr{T}}^{\lambda_{k}})^*
  \mathscr{T}=\mathscr{T}^*\sum_{k=1}^NC_{\mathscr{T}}^{\lambda_{k}}
  (C_{\mathscr{T}}^{\lambda_{k}})^*\mathscr{T}\,.
\end{equation*}
\end{proof}
The next assertion is an immediate consequence of the previous result.
\begin{corollary}
  \label{cor:c-properties}
  The following property of the vectors
  $C^{\lambda_{k}}_{\mathscr T}$, $k=1,\dots,N$ takes place
  \begin{equation*}
    \begin{pmatrix}C^{\lambda_{1}}_{\mathscr T}&\dots &
                                                        C^{\lambda_{N}}_{\mathscr T}\end{pmatrix}^{*}\widetilde\delta_{j}\ne 0\,
                                                      \text{ for all }\, j=1,\dots, n\,,
  \end{equation*}
  where $\widetilde\delta_j=\widetilde\Pi_{\nats\to G_n}\delta_j$.
\end{corollary}
\begin{proof}
  If there is $j_{0}\in G_{n}$ such that
  \begin{equation*}
  \begin{pmatrix}C_{\mathscr T}^{\lambda_{1}}&\dots &
                                      C_{\mathscr T}^{\lambda_{N}}\end{pmatrix}^{*}\widetilde\delta_{j_{0}}= 0,
  \end{equation*}
  then the matrix
  $C_{\mathscr{T}}^{\lambda_{k}} (C_{\mathscr{T}}^{\lambda_{k}})^*$
  would have a zero row (the $j_{0}$-th one) for all $k\in G_{N}$,
  which contradicts Lemma~\ref{lem:restricted-identity} in view of the
  properties of $\mathscr T$.
\end{proof}
\begin{definition}
  \label{def:spectral-function}
  Define the matrix-valued function
  \begin{equation}
    \label{eq:spec-function}
    \sigma_{N}^{\mathscr
      T}(t):=\sum_{\lambda_k<t}C_{\mathscr{T}}^{\lambda_{k}}
    (C_{\mathscr{T}}^{\lambda_{k}})^*\,,\quad t\in\reals\,,
     \end{equation}
     and the space of $n$-dimensional vector valued functions
     $L_{2}(\reals,\sigma_{N}^{\mathscr T})$ with inner product
\begin{equation}
  \label{eq:inner-product-finite}
  \inner{\pb f}{\pb g}_{L_{2}(\reals,\sigma_{N}^{\mathscr T})}:=
  \int_{\reals}\inner{\pb f(t)}{d\sigma_{N}^{\mathscr T}(t)\pb g(t)}_{l_{2}(G_{n})}\,.
\end{equation}
\end{definition}
Note that $\sigma_{N}^{\mathscr T}$ has $N$ points of growth so
$L_{2}(\reals,\sigma_{N}^{\mathscr T})$ is an $N$-dimensional space
and the inner product can be expressed as follows
\begin{equation}
  \label{eq:inner-product-other}
  \inner{\pb f}{\pb g}_{L_{2}(\reals,\sigma_{N}^{\mathscr T})}=
  \sum_{k=1}^{N}\pb f(\lambda_{k})^{*}C_{\mathscr{T}}^{\lambda_{k}}
  (C_{\mathscr{T}}^{\lambda_{k}})^*\pb g(\lambda_{k})\,.
\end{equation}
   \begin{definition}
     \label{def:p_q_poly}
     Define the $n$-dimensional vector polynomials:
     \begin{align}
       \label{eq:p_q_poly}
       \pb p_{k}(z)&:=\Psi_{\mathscr T}(z)^{*}\delta_k\,, \quad
                     k\in G_{N}\\
       \label{eq:q_q_poly}
 \pb q_{j}(z)&:=\Theta_{\mathscr T}(z)^{*}\delta_j\,, \quad j\in G_{n}\,,
     \end{align}
     where $\Psi_{\mathscr T}(z)$ and $\Theta_{\mathscr T}(z)$ are
     given in \eqref{eq:notation-colums-composition} and
     \eqref{eq:Q_matrix}, respectively.
   \end{definition}
   \begin{remark}
     It follows from Definition~\ref{def:p_q_poly} that
     \begin{equation}
       \label{eq:q_relation_p}
       \begin{pmatrix}
  \pb q_1&\dots&\pb q_n
\end{pmatrix}^*
   =\widetilde\Pi_{G_{N}\to K}(M_{N}-zI)
\begin{pmatrix}
  \pb p_1&\dots&\pb p_N
\end{pmatrix}^*\,.
\end{equation}
and
    \begin{equation}
       \label{eq:p_relation}
       \begin{pmatrix}
  \pb p_1&\dots&\pb p_N
\end{pmatrix}^*
   (M_{N}-\overline{z}I)\widetilde\Pi_{G_{N}\to K^{\perp}}^{*}=0
\end{equation}
\end{remark}
For Example~\ref{ex:simple-difference-equation}, the expression
\eqref{eq:q_relation_p} gives the following equalities:
\begin{align*}
  \pb{q}_1(z)&= (m_{33}-z)\pb{p}_3+m_{31}\pb{p}_1+m_{32}\pb{p}_2+m_{34}\pb{p}_4+m_{35}\pb{p}_5+m_{36}\pb{p}_6\\
  \pb{q}_2(z)&=(m_{55}-z)\pb{p}_5+m_{52}\pb{p}_2+m_{53}\pb{p}_3+m_{54}\pb{p}_4+m_{56}\pb{p}_6\\
  \pb{q}_3(z)&=(m_{77}-z) \pb{p}_7+m_{76}\pb{p}_6
\end{align*}
and \eqref{eq:p_relation} can be solved recursively as follows
\begin{align*}
  \pb{p}_1(z)&=t_{11}\widetilde\delta_1\\
  \pb{p}_2(z)&=t_{12}\widetilde\delta_1+t_{22}\widetilde\delta_2\\
  \pb{p}_3(z)&=t_{13}\widetilde\delta_1+t_{23}\widetilde\delta_2+t_{33}\widetilde\delta_3\\
  \pb{p}_4(z)&=\pb{m}_{14}^{-1}\left[(z-m_{11})\pb{p}_1(z)-m_{12}\pb{p}_2(z)-m_{13}\pb{p}_3(z)\right]\\
             &\vdots\\
  \pb{p}_7(z)&=\pb{m}_{67}^{-1}\left[(z-m_{66})\pb{p}_6(z)-m_{62}\pb{p}_2(\zeta)-m_{63}\pb{p}_3(z)-m_{65}\pb{p}_5(z)\right]\,,
\end{align*}
where $\widetilde{\delta}_{j}$ is given in
Corollary~\ref{cor:c-properties}.
\begin{proposition}
  \label{prop:ortonormal-p-L2-finite}
The vector polynomials
  $\{\boldsymbol{p}_k(z)\}_{k=1}^{N}$, defined by
  (\ref{eq:p_q_poly}), satisfy
\begin{equation*}
  \inner{\boldsymbol{p}_j}{\boldsymbol{p}_k}
_{L_2(\mathbb{R},\sigma_N^{\mathscr{T}})}
=\delta_{jk}\quad\text{for}\ j,k\in G_{N},
\end{equation*}
where $\delta_{jk}$ is the Kronecker symbol.
\end{proposition}
\begin{proof}
  Using \eqref{eq:inner-product-other} and \eqref{eq:p_q_poly}, one
  obtains
  \begin{align*}
    \inner{\pb{p}_i}{\pb{p}_j}_{L_2(\reals,\sigma_N^{\mathscr{T}})}&=
                                                                     \sum_{k=1}^N(\pb{p}_i(\lambda_k))^*C_{\mathscr T}^{\lambda_{k}}(C_{\mathscr T}^{\lambda_{k}})^*\pb{p}_j(\lambda_k)\\
                                                                   &=\sum_{k=1}^N\delta_i^*\Psi_{\mathscr{T}}(\lambda_k)C_{\mathscr T}^{\lambda_{k}}(C_{\mathscr T}^{\lambda_{k}})^*\Psi_{\mathscr{T}}(\lambda_k)\\
                                                                   &=\sum_{k=1}^N\delta_i^*\varphi^{\lambda_{k}}(\varphi^{\lambda_{k}})^*\delta_j=\delta_i^*\delta_j=\delta_{ij}\,,
  \end{align*}
  where in the third equality, one recurs to
  \eqref{eq:relation_phi_psi}.
\end{proof}
\begin{corollary}
  \label{cor:isometry-finite}
  The function $\Psi_{\mathscr T}(t)$ (see \eqref{eq:psi_tau}) with
  $t\in\reals$ gives rise to a map $\Upsilon_{\mathscr T}$ by means of
  the expression
 \begin{equation*}
   l_{2}(G_{N})\ni u \stackrel{\Upsilon_{\mathscr T}}{\longmapsto}  \Psi^{*}(t)u\in L_2(\mathbb{R},\sigma_N^{\mathscr{T}})\,.
  \end{equation*}
This map is an isometry.
\end{corollary}
\begin{proof}
  Indeed, by \eqref{eq:p_q_poly}, $\Upsilon_{\mathscr T}$ maps the
  canonical basis into the orthonormal basis $\{\pb p_{k}\}_{k=1}^{N}$
  in $L_2(\mathbb{R},\sigma_N^{\mathscr{T}})$.
\end{proof}
\begin{proposition}
  \label{prop:multiplication-operator}
  The isometry $\Upsilon_{\mathscr T}$ transforms the operator
  $\mathcal M_{N}$ into the operator of multiplication by the
  independent variable in $L_2(\mathbb{R},\sigma_N^{\mathscr{T}})$.
\end{proposition}
\begin{proof} Taking into account that
  $\mathcal
  M_{N}=\sum_{k=1}^{N}\lambda_k\varphi^{\lambda_{k}}(\varphi^{\lambda_{k}})^*$
  together with \eqref{eq:relation_phi_psi} and
  \eqref{eq:inner-product-other}, one obtains
\begin{align*}
  \inner{\delta_i}{\mathcal M_{N}\delta_j}_{l_{2}(G_{N})}&=\delta_{i}^{*}\sum_{k=1}^{N}\lambda_k\varphi^{\lambda_{k}}(\varphi^{\lambda_{k}})^*\delta_j\\
                                          &=\sum_{k=1}^{N}\lambda_k\delta_i^*\Psi C_{\mathscr T}^{\lambda_{k}}(C_{\mathscr T}^{\lambda_{k}})^*\Psi^*\delta_j\\
                                          &=\sum_{k=1}^{N}\lambda_k(\pb p_{i}(\lambda_k))^* C_{\mathscr T}^{\lambda_{k}}(C_{\mathscr T}^{\lambda_{k}})^*\pb p_{j}(\lambda_k)\\
                                          &=\inner{\pb p_{i}(t)}{t\pb p_{j}(t)}_{L_2(\reals,\sigma_N^{\mathscr{T}})}\,.
\end{align*}
\end{proof}
\begin{remark}
  \label{rem:spectral-theorem}
  In view of the spectral theorem,
  Proposition~\ref{prop:multiplication-operator} is somehow
  straightforward.  However, the concrete realization for the boundary
  matrix $\mathscr{T}$ together with the construction algorithm are
  relevant for the further discussion and elaborations.
\end{remark}
\begin{proposition}
  \label{prop:zeros-finite}
  The vector polynomials $\{\pb q_{j}\}_{j=1}^{n}$, defined in
  \eqref{eq:q_q_poly}, have zero norm in
  $L_2(\reals,\sigma_N^{\mathscr{T}})$.
\end{proposition}
\begin{proof}
  By \eqref{eq:inner-product-other} and \eqref{eq:q_q_poly}, one has
\begin{align*}
  \inner{\boldsymbol{q}_{j}}{\boldsymbol{q}_{j}}_{L_2(\mathbb{R},\sigma_N^{\mathscr{T}})}
&=\sum_{k=1}^{N}\pb q_{j}(\lambda_{k})^{*}C_{\mathscr T}^{\lambda_{k}}(C_{\mathscr
  T}^{\lambda_{k}})^{*}\pb q_{j}(\lambda_{k})\\
  &= \sum_{k=1}^{N}\delta_{j}^{*}\Theta_{\mathscr T}(\lambda_{k})C_{\mathscr T}^{\lambda_{k}}(C_{\mathscr
  T}^{\lambda_{k}})^{*}\Theta_{\mathscr T}(\lambda_{k})^{*}\delta_{j}\,.
\end{align*}
The assertion then follows after noticing that
\begin{equation*}
  \Theta_{\mathscr T}(\lambda_{k})C_{\mathscr T}^{\lambda_{k}}=\widetilde\Pi_{G_{n}\to K}(M_{N}-\lambda_{k}I) \Psi_{\mathscr T}(\lambda_{k})C_{\mathscr T}^{\lambda_{k}}=0\,,
\end{equation*}
due to the fact that $\Psi_{\mathscr T}(\lambda_{k})C_{\mathscr T}^{\lambda_{k}}=\varphi^{\lambda_{k}}$.
\end{proof}
\begin{theorem}
  \label{thm:spectral-measure-finite}
  The function $\sigma_N^{\mathscr{T}}(t)$ given in
  Definition~\ref{def:spectral-function} is a spectral function of
  $\mathcal M_{N}$, that is a complete measure\footnote{Completeness
    of a spectral measure is given in \cite[Chp.\,5
    Sec.\,1.1]{MR1192782}. Here it means that
    $\mathscr T^{*}\sum_{k=1}^{N}C_{\mathscr
      T}^{\lambda_{k}}(C_{\mathscr T}^{\lambda_{k}})^{*}\mathscr T=I$,
    see \cite[Lem.\,2.3]{MR3711273}.}  with the following properties:
\begin{enumerate}[(I)]
\item It is a nondecreasing monotone step function which is continuous
  from the left.
\item Each jump is a matrix of rank not greater than $n$.
\item The sum of the ranks of all jumps equals $N$.
\end{enumerate}
\end{theorem}
\begin{proof}
  The fact that $\sigma_N^{\mathscr{T}}$ is a spectral function for
  $\mathcal M_{N}$ is a consequence of
  Corollary~\ref{cor:isometry-finite} and
  Proposition~\ref{prop:multiplication-operator}. Note that
  completeness of the spectral measure is also given by
  Proposition~\ref{prop:multiplication-operator}. The property (I) is
  an immediate outcome of Definition~\ref{def:spectral-function} since
  the quadratic form
  $\pb{f}^{*}C_{\mathscr T}^{\lambda_{k}}(C_{\mathscr
    T}^{\lambda_{k}})^*\pb f$ is the norm squared of the vector
  $(C_{\mathscr T}^{\lambda_{k}})^*\pb f$. To prove (II) and (III),
  one uses Corollary~\ref{cor:eigenspaces} and the fact that
  $C_{\mathscr T}^{\lambda_{k}}(C_{\mathscr T}^{\lambda_{k}})^*$ is a
  rank one matrix since $C_{\mathscr T}^{\lambda_{k}}$ is a nonzero
  vector.
\end{proof}

\begin{remark}
  \label{rem:equivalence-measures}
  For any complete function $\sigma$ satisfying (I)--(III) of
  Definition~\ref{def:spectral-function} and having jumps in
  $\mu_{1},\dots,\mu_{N}$, there are $n$-dimensional vectors
  $C^{1},\dots,C^{N}$ such that
\begin{equation}
  \label{eq:interpolation-problem-sigma}
  \sigma(t)=\sum_{\mu_{k}<t}C^{k}(C^{k})^{*}\,,
\end{equation}
where there are no zero elements in the collection $C^{1},\dots,C^{N}$
and it satisfies Corollary~\ref{cor:c-properties} with
$C^{\lambda_{k}}_{\mathscr T}=C^{k}$ for any $k=1,\dots,N$ (see
\cite[Thm.\,2.2]{MR1668981} and \cite[Sec.\,2]{MR3711273}).  As has
been shown in this section, any such collection of vectors
$C^{1},\dots,C^{N}$ and numbers $\mu_{1},\dots,\mu_{N}$ determine
uniquely by \eqref{eq:interpolation-problem-sigma} a complete function
which satisfies (I)--(III).
\end{remark}
\section{Band diagonal matrices with degenerations}
\label{sec:band-diag-matr}

\begin{definition}
  \label{def:ourclass}
  The matrix $M\in\mathfrak M$ (see Definition~\ref{def:class}) is in
  the class $\widetilde{\mathfrak M}$ when the following conditions
  are met.
    \begin{enumerate}[(a)]
    \item If $(k(j),n+j)$ are the coordinates of the row-edge entries
      for $j\in\nats$, then $k(j)<k(j+1)$ for $j\in\nats$.
    \item All the row-edge entries are simultaneously column-edge
      entries.
    \end{enumerate}
\end{definition}
\begin{remark}
  \label{rem:subclass}
  Note that item $(a)$ in Definition~\ref{def:ourclass} allows for the
  existence of rows without row-edge entries. We say that there is a
  \emph{degeneration} in each such row (see
  Fig.~\ref{fig:structure-class}). It is straightforward to verify
  that the class $\widetilde{\mathfrak M}$ is equivalent to the class
  given in \cite[Def.\,1]{MR3543793}.
\end{remark}

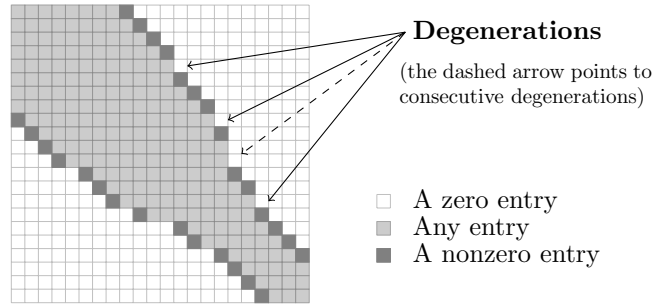
\begin{figure}[h]
\begin{center}
\begin{tikzpicture}[scale=.18]\footnotesize
  \pgfmathsetmacro{\xone}{0} \pgfmathsetmacro{\xtwo}{22}
  \pgfmathsetmacro{\yone}{8} \pgfmathsetmacro{\ytwo}{30}
  \draw[step=1cm,gray,opacity=0.5,very thin] (\xone,\yone) grid
  (\xtwo,\ytwo); 
  \draw[step=1cm,gray,opacity=0.5,very thin] (27,15)
  grid (28,16); 
  \draw[step=1cm,gray,opacity=0.5,very thin] (27,13) grid
  (28,14); 
  \draw[step=1cm,gray,opacity=0.5,very thin] (27,11) grid
  (28,12); 
  \draw(38,15.5)node[text width=3cm,scale=1]{A zero entry};
  \draw(38,13.5)node[text width=3cm,scale=1]{Any entry};
  \draw(38,11.5)node[text width=3cm,scale=1]{A nonzero entry};
  \draw(38,28)node[text width=3cm,scale=1]{\bf Degenerations};
  \draw(38,24.1)node[text width=4.2cm,align=justify, scale=.8] {(the
    dashed arrow points to consecutive degenerations)}; \draw[->]
  (29.1,28) -- (16,21.5);
  \draw[->] (29.1,28) -- (13,25.5); 
  \draw[dashed,->] (29.1,28) -- (17,19); 
  \draw[->] (29.1,28) -- (19,15.5); 


\begin{scope}
  \filldraw[thin,gray,opacity=.4] (27,13)
    rectangle (28,14)
 ;
  \filldraw[thin,gray,opacity=1] (27,11)
    rectangle (28,12);
\end{scope}
\begin{scope}
\foreach \x in {0,1,2,3}
{
  \filldraw[thin,gray,opacity=1] (0+\x, 21-\x)
    rectangle (1+\x,22-\x)
 ;
   \filldraw[thin,gray,opacity=1] (8+\x, 30-\x)
     rectangle (9+\x,29-\x);}
\end{scope}
\begin{scope}
\foreach \x in {0,1,2,3,4}
{
  \filldraw[thin,gray,opacity=.4] (0+\x, 22-\x)
    rectangle (1+\x,23-\x)
 ;
   \filldraw[thin,gray,opacity=.4] (7+\x, 30-\x)
     rectangle (8+\x,29-\x);}
\end{scope}

\begin{scope}
\foreach \x in {5,6,7}
{
  \filldraw[thin,gray,opacity=1] (0+\x, 22-\x)
    rectangle (1+\x,23-\x)
 ;
   \filldraw[thin,gray,opacity=1] (7+\x, 30-\x)
     rectangle (8+\x,29-\x);}
\end{scope}
\begin{scope}
\foreach \x in {0,1,2,3,4,5,6,7,8}
{
  \filldraw[thin,gray,opacity=.4] (0+\x, 23-\x)
    rectangle (1+\x,24-\x)
 ;
   \filldraw[thin,gray,opacity=.4] (6+\x, 30-\x)
     rectangle (7+\x,29-\x);}
\end{scope}
\begin{scope}
\foreach \x in {9}
{
  \filldraw[thin,gray,opacity=1] (0+\x, 23-\x)
    rectangle (1+\x,24-\x)
 ;
   \filldraw[thin,gray,opacity=1] (6+\x, 30-\x)
     rectangle (7+\x,29-\x);}
\end{scope}
\begin{scope}
\foreach \x in {0,1,2,3,4,5,6,7,8,9,10}
{
  \filldraw[thin,gray,opacity=.4] (0+\x, 24-\x)
    rectangle (1+\x,25-\x)
 ;
   \filldraw[thin,gray,opacity=.4] (5+\x, 30-\x)
     rectangle (6+\x,29-\x);}
\end{scope}

\begin{scope}
\foreach \x in {0,1,2,3,4,5,6,7,8,9,10,11,12}
{
  \filldraw[thin,gray,opacity=.4] (0+\x, 25-\x)
    rectangle (1+\x,26-\x)
 ;
   \filldraw[thin,gray,opacity=.4] (4+\x, 30-\x)
     rectangle (5+\x,29-\x);}
\end{scope}

\begin{scope}
\foreach \x in {12,13}
{
  \filldraw[thin,gray,opacity=1] (0+\x, 25-\x)
    rectangle (1+\x,26-\x)
 ;
   \filldraw[thin,gray,opacity=1] (4+\x, 30-\x)
     rectangle (5+\x,29-\x);}
\end{scope}
\begin{scope}
\foreach \x in {0,1,2,3,4,5,6,7,8,9,10,11,12,13,14}
{
  \filldraw[thin,gray,opacity=.4] (0+\x, 26-\x)
    rectangle (1+\x,27-\x)
 ;
   \filldraw[thin,gray,opacity=.4] (3+\x, 30-\x)
     rectangle (4+\x,29-\x);}
\end{scope}

\begin{scope}
\foreach \x in {15,16,17,18}
{
  \filldraw[thin,gray,opacity=1] (0+\x, 26-\x)
    rectangle (1+\x,27-\x)
 ;
   \filldraw[thin,gray,opacity=1] (3+\x, 30-\x)
     rectangle (4+\x,29-\x);}
\end{scope}

\begin{scope}
\foreach \x in {0,1,2,3,4,5,6,7,8,9,10,11,12,13,14,15,16,17,18,19}
{
  \filldraw[thin,gray,opacity=.4] (0+\x, 27-\x)
    rectangle (1+\x,28-\x)
 ;
   \filldraw[thin,gray,opacity=.4] (2+\x, 30-\x)
     rectangle (3+\x,29-\x);}
\end{scope}

\begin{scope}
\foreach \x in {0,1,2,3,4,5,6,7,8,9,10,11,12,13,14,15,
16,17,18,19,20}
{
  \filldraw[thin,gray,opacity=.4] (0+\x, 28-\x)
    rectangle (1+\x,29-\x)
 ;
   \filldraw[thin,gray,opacity=.4] (1+\x, 30-\x)
     rectangle (2+\x,29-\x);}
\end{scope}

\begin{scope}
  \foreach \x in
  {0,1,2,3,4,5,6,7,8,9,10,11,12,13,14,15,16,17,18,19,
    20,21}
  { \filldraw[thin,gray,opacity=.25] (0+\x, 29-\x) rectangle
    (1+\x,30-\x) ; \filldraw[thin,gray,opacity=.2] (0+\x, 29-\x)
    rectangle (1+\x,30-\x);}
\end{scope}
\end{tikzpicture}
\end{center}
\caption{The structure of a matrix in
  $\widetilde{\mathfrak M}$}\label{fig:structure-class}
\end{figure}

In \cite[Thm.\,4.1]{MR3711273} the next result is proven. It plays an
important role in our further considerations.
\begin{theorem}
  \label{thm:finite-our-class-inversion}
  If $\sigma$ is a $n\times n$-matrix-valued function defined on the
  real line and satisfying (I)--(III) of
  Theorem~\ref{thm:spectral-measure-finite}, then there is an
  upper-triangular invertible $n\times n$-matrix $\mathscr T$ and a
  matrix $\widetilde M\in\widetilde{\mathfrak M}$ such that
  $\sigma=\widetilde{\sigma}_{N}^{\mathscr T}$. Here
  $\widetilde{\sigma}_{N}^{\mathscr T}$ is the function given by
  Definition~\ref{def:spectral-function} for the matrix
  $\widetilde{M}_{N}$.
\end{theorem}
\begin{remark}
  \label{rem:reconstruction-algorithm}
  The results of \cite[Sec.\,4]{MR3711273} give a constructive
  algorithm for finding the entries of $\widetilde{M}_{N}$ on the
  basis of the function $\sigma$.
\end{remark}
\begin{theorem}
  \label{thm:equivalence-classes}
  For any $M\in\mathfrak M$ there is a
  $\widetilde M\in\widetilde{\mathfrak M}$ such that
  $n_{M}=n_{\widetilde M}$ (see Definition~\ref{def:class}) and
  $M_{N}$ is unitary equivalent to $\widetilde{M}_{N}$ for any
  $N>n=n_{M}=n_{\widetilde M}$. Here we actually mean that
  $\mathcal{M}_{N}$ is unitary equivalent to
  $\widetilde{\mathcal M}_{N}$ (see the last paragraph of
  Section~\ref{sec:preliminaries}).
\end{theorem}
\begin{proof}
  For a matrix $M\in\mathfrak M$, cosider the
  $n_{M}\times n_{M}$-matrix-valued function $\sigma_{N}^{\mathscr T}$
  given in Definition~\ref{def:spectral-function} with $N>n_{M}$ and
  an arbitrary boundary matrix $\mathscr T$. On the basis of
  Theorems~\ref{thm:spectral-measure-finite} and
  \ref{thm:finite-our-class-inversion} there is an upper-triangular
  invertible matrix $\widetilde{\mathscr T}$ and a matrix
  $\widetilde M\in\widetilde{\mathfrak M}$ such that
  $\sigma_{N}^{\mathscr T}=\widetilde{\sigma}_{N}^{\widetilde{\mathscr
      T}}$, where $\widetilde{\sigma}_{N}^{\widetilde{\mathscr T}}$ is
  the spectral function of $\widetilde{M}_{N}$. Clearly,
  $n_{M}=n_{\widetilde M}$. For completing the proof, use
  Proposition~\ref{prop:multiplication-operator}.
\end{proof}

The interpolation problem for $n$-dimensional vector polynomials as
stated in \cite{MR3389906} consists in finding $n$-dimensional vector
polynomials $\pb{r}(z)$ such that
\begin{equation}
  \label{eq:interpolation-problem}
\pb{r}(\mu_{k})^*C^{k}(C^{k})^{*}\pb{r}(\mu_{k})=0\,
\end{equation}
for $k=1,\dots,N$. Here $C^{1},\dots, C^{N}$ is a collection of
nonzero $n$-dimensional vectors and $\mu_{1},\dots,\mu_{N}$ is a
collection of real numbers in which each number can appear repeatedly
at most $n$ times.  This problem is a generalization of the rational
interpolation problem (also known as Cauchy-Jacobi problem) studied in
\cite{MR1091797} and used by \cite{MR2533388} for the spectral
analysis of CMV matrices.

In view of the second part of Remark~\ref{rem:equivalence-measures}
and Theorem~\ref{thm:finite-our-class-inversion}, one can always
consider the collection of numbers appearing in
Remark~\ref{rem:equivalence-measures} to be
$\lambda_{1},\dots,\lambda_{N}$, \emph{i.\,e.} the spectrum of
$\widetilde{M}_{N}$ ($\widetilde M\in\widetilde{\mathfrak M}$), and
the collection of vectors to be the corresponding
$C_{\mathscr T}^{\lambda_{1}},\dots, C_{\mathscr T}^{\lambda_{N}}$ for
a boundary matrix $\mathscr T$.  Due to the fact that
$C_{\mathscr T}^{\lambda_k}(C_{\mathscr T}^{\lambda_k})^{*}$ is
nonnegative, the interpolation problem is equivalent to finding vector
polynomials in the equivalence class of the zero function in
$L_{2}(\reals,\sigma_{N}^{\mathscr T})$.
\begin{definition}
\label{def:height}
Let
$\boldsymbol{r}(z)=\left(r_1(z),r_2(z),\ldots,r_n(z)\right)^{\intercal}$
be an $n$-dimensional vector polynomial. The height of
$\boldsymbol{r}(z)$ is
\begin{equation*}
h(\boldsymbol{r}):=
\max_{j\in\{1,\dots,n\}}\left\lbrace n\deg (r_j)+j-1\right\rbrace\,,
\end{equation*}
where it is assumed that $\deg 0:=-\infty$ and
$h(\boldsymbol{0}):=-\infty$.  The height of the set $S$ is defined by
\begin{equation*}
  h(S):= \min\{h(\pb{r}): \pb{r}\in S, \pb{r}\ne 0\}\,.
\end{equation*}
\end{definition}
\begin{definition}
  \label{def:generators}
  Denote by $\mathbb{S}$ the solutions to the interpolation problem,
  that is, the set of all $n$-dimensional vector polynomials
  satisfying \eqref{eq:interpolation-problem} for $k=1,\dots,N$.
  Consider $\pb{g}_{1}\in\mathbb{S}$ such that
  $h(\pb{g}_{1})=h(\mathbb{S})$ and
  $\pb{g}_{k}\in\mathbb{S}\setminus\Span\{M(\pb{g}_{j})\}_{j=1}^{k-1}$
  such that
  \begin{equation*}
    h(\pb{g}_{k})=h(\mathbb{S}\setminus
    \Span\{\mathbb{M}(\pb{g}_{j})\}_{j=1}^{k-1})\,,
  \end{equation*}
  where
  $\mathbb{M}(\pb{r}):=\{\pb s: \pb s(z)=s(z)\pb{r}(z), s\ \text{is an
    arbitrary scalar polynimial}\}$. The vector polynomial
  $\pb{g}_{k}$ is the $k$-th generator of $\mathbb S$.
\end{definition}

If there are two elements of $\mathbb{S}$ with the same height, then
they differ from each other only by a multiplicative constant
\cite[Lem.\, 4.1]{MR3389906}. On the basis of this, it is shown in
\cite[Sec\,4]{MR3389906} (see also \cite[Prop.\,3.3]{MR3711273}) that
the generators are unique modulo a multiplicative constant and that
\begin{equation*}
  \mathbb{S}=\mathbb{M}(\pb{g_{1}})\dotplus\dots\dotplus\mathbb{M}(\pb{g_{n}})\,,
\end{equation*}
where $\dotplus$ denotes \emph{direct sum}. Note that the generators
$\pb{g}_{1},\dots,\pb{g}_{n}$ have different heights and the height of
any element in $\mathbb{M}(\pb{g}_{i})$ is $h(\pb{g}_{i})+kn$ with
$k\in\nats\cup\{0\}$.

According to Proposition~\ref{prop:zeros-finite},
$\pb{q}_{j}\in\mathbb{S}$ for all $j=1,\dots,n$. Moreover, the
following assertion is proven in \cite[Thm.\,3.1]{MR3711273}.
\begin{proposition}
  \label{prop:q-generators}
  If $M\in\widetilde{\mathfrak{M}}$, then $\pb{q}_{j}$ is the $j$-th
  generator of $\mathbb{S}$ for $j=1,\dots,n$.
\end{proposition}
\begin{remark}
  \label{rem:generators-restatement}
  A way of restating the previous proposition is to say that, for
  $M\in\widetilde{\mathfrak{M}}$, $\pb{q}_{j}=\pb{g}_{j}$ (modulo a
  multiplicative constant) for $j=1,\dots,n$. This implies that, when
  $M\in\widetilde{\mathfrak{M}}$, the heights of
  $\pb{q}_{1},\dots,\pb{q}_{n}$ are in different equivalent classes of
  $\integers/n\integers$.
\end{remark}
\begin{remark}
  \label{rem:q-marchenko-not}
  It is not true in general that, for $M\in\mathfrak{M}$, $\pb{q}_{j}$
  is the $j$-th generator of $\mathbb{S}$. This is a remarkable
  difference between the classes $\mathfrak{M}$ and
  $\widetilde{\mathfrak{M}}$, actually $\widetilde{\mathfrak{M}}$ was
  somehow tailored to have the property given by
  Proposition~\ref{prop:q-generators}.
\end{remark}

\begin{proposition}
  \label{prop:our-class-heights}
  Let $M\in\widetilde{\mathfrak M}$ and $\pb{p}_{k}$ and $\pb{q}_{j}$
  be the corresponding vector polynomials given by \eqref{eq:p_q_poly}
  and \eqref{eq:q_q_poly}, respectively. The heights of the
  polynomials $\pb{p}_{1},\dots,\pb{p}_{N}$ are increasing and cannot
  coincide with the heights of any polynomial in $\mathbb{S}$ (see
  \cite[Lem.\,3.2]{MR3711273}).
\end{proposition}
\begin{remark}
  \label{rem:no-heights}
  For a matrix $M\in\mathfrak{M}$,
  Proposition~\ref{prop:our-class-heights} does not hold. There are
  cases (see Example~\ref{ex:missing-heights} below) where the heights
  of the polynomials $\pb{p}_{1},\dots,\pb{p}_{N}$ are not increasing
  or coincide with the heights of elements in $\mathbb{S}$.
\end{remark}
There is a relavant property for the heights of the vector polynomials
considered in the hypothesis of
Proposition~\ref{prop:our-class-heights}. Indeed, it turns out that
the set of numbers
$h(\pb{p}_{1}),\dots,h(\pb{p}_{N}),
h(\pb{q}_{1})+nk_{1},\dots,h(\pb{q}_{n})+nk_{n} $ with
$k_{1},\dots,k_{n}\in\nats\cup\{0\}$ cover all nonnegative integers.
As a consequence of this and the fact that any nonzero vector
polynomial of height $h$ is decomposed in a linear combination of
vector polynomials of heights $0,\dots,h$
\cite[Prop.\,3.1]{MR3711273}, one has the following proposition (see
\cite[Cor.\,3.2]{MR3711273}).
\begin{proposition}
  \label{prop:decomposition-vector-polynomials}
  Let $M\in\widetilde{\mathfrak M}$ and $\pb{p}_{k}$ and $\pb{q}_{j}$
  be the corresponding vector polynomials given by \eqref{eq:p_q_poly}
  and \eqref{eq:q_q_poly}, respectively.  Any vector polynomial
  $\pb{r}$ admits the decomposition
  \begin{equation*}
    \pb{r}(z)=\sum_{k=1}^{N}a_{k}\pb{p}_{k}(z)+\sum_{j=1}^{n}s_{j}(z)\pb{q}_{j}(z)\,,
  \end{equation*}
  where $a_{k}\in\complex$ and $s_{j}(z)$ is a scalar polynomial for
  any $k=1,\dots,N$ and $j=1,\dots,n$.
\end{proposition}
\begin{example}
  \label{ex:missing-heights}
  Consider Example~\ref{ex:simple-difference-equation} again.
  \begin{enumerate}
  \item If $m_{25}=0$ and $m_{35}\ne 0$, then
    $h(\pb{p}_{5})=h(\pb{q}_{1})=6$ (see Remark~\ref{rem:no-heights})
    and $h(\pb{q}_{2})=9$ which is in the same equivalence class of
    $\integers/3\integers$ as $h(\pb{q}_{1})$ (see
    Remark~\ref{rem:q-marchenko-not}).
  \item If $m_{25}\ne 0$ and $m_{35}=0$, then
    $h(\pb{p}_{5})=h(\pb{p}_{6})=h(\pb{q}_{1})=6$,
    $h(\pb{p}_{7})=h(\pb{q}_{2})=9$ $h(\pb{q}_{3})=12$. See
    Remarks~\ref{rem:q-marchenko-not} and \ref{rem:no-heights} and
    note that in this case $h(\pb{q}_{1})$, $h(\pb{q}_{2})$ and
    $h(\pb{q}_{3})$ all fall in the same equivalence class of
    $\integers/3\integers$.
  \item If $m_{25}=0$ and $m_{35}= 0$, then the set of numbers
    $h(\pb{p}_{1}),\dots,h(\pb{p}_{7}),
    h(\pb{q}_{1})+3k_{1},\dots,h(\pb{q}_{3})+3k_{3} $ with
    $k_{1},\dots,k_{3}\in\nats\cup\{0\}$ cover all nonnegative
    integers. Here $h(\pb{q}_{1})$, $h(\pb{q}_{2})$ and
    $h(\pb{q}_{3})$ fall in the three different equivalence classes of
    $\integers/3\integers$, but the sequence
    $h(\pb{p}_{1}),\dots,h(\pb{p}_{7})$ is not increasing.
  \end{enumerate}
\end{example}

\section{Spectral functions for the class $\mathfrak M$}
\label{sec:infin-march-slav}

The vector polynomials given in Definition~\ref{def:p_q_poly} for a
matrix $M\in\mathfrak{M}$ can be used in a decomposition similar to
the one given in
Proposition~\ref{prop:decomposition-vector-polynomials} which was
restricted to the class $\widetilde{\mathfrak{M}}$. To prove this
decomposition for the general case, the following simple observation
is needed.
\begin{lemma}
  \label{lem:simple-heights}
  Let $M\in\mathfrak{M}$ and $N>n$. Consider that $n$-dimensional
  vector polynomials $\pb{p}_{k}$ given by \eqref{eq:p_q_poly}.  For
  any given $h\in\nats\cup\{0\}$, there are $l\in\nats\cup\{0\}$ and
  $m\in G_{n}$ (see last paragraph of Section~\ref{sec:preliminaries})
  such that
  \begin{equation*}
    h(z^{l}\pb{p}_{m}(z))=h\,.
  \end{equation*}
\end{lemma}
\begin{proof}
  Write $h=nl+k-1$, where $l\in\nats\cup\{0\}$ and $k\in G_{n}$.  By
  \eqref{eq:p_q_poly}, $ h(\pb{p}_{k}(z))=k-1$ for $k\in G_{n}$. On
  the other hand, it follows from \cite[Eq. 3]{MR3389906} that
  \begin{equation*}
     h(z^{l}\pb{p}_{k}(z))=nl+k-1\quad\text{for any } k\in G_{n}\,.
  \end{equation*}
\end{proof}
The following assertion plays an important role in this section.
\begin{theorem}
  \label{thm:decomposition-marchenko}
  Let $M\in\mathfrak{M}$ and $N>n$. For any $l\in\nats\cup\{0\}$ and
  $m\in G_{n}$, there are $a_{k}\in\complex$ ($k\in G_{N}$) and scalar
  polynomials $S_{j}$ ($j\in G_{n}$) such that
 \begin{equation*}
   z^{l}\pb{p}_{m}(z)=\sum_{k=1}^{N}a_{k}\pb{p}_{k}(z)+\sum_{j=1}^{n}S_{j}(z)\pb{q}_{j}(z)\,,
 \end{equation*}
 where $\pb{p}_{k}$ and $\pb{q}_{j}$ are given by \eqref{eq:p_q_poly}
 and \eqref{eq:q_q_poly}.
\end{theorem}
\begin{proof}
  The assertion follows trivially for $l=0$. When $l=1$, it is a
  consequence of \eqref{eq:q_relation_p} and \eqref{eq:p_relation}. If
  the statement holds for $l$, then
  \begin{align*}
    z^{l+1}\pb{p}_{m}(z)&=\sum_{k=1}^{N}a_{k}z\pb{p}_{k}(z)+
                          \sum_{j=1}^{n}zS_{j}(z)\pb{q}_{j}(z)\\
                        &=\left(\sum_{k\in K^{\perp}} +\sum_{k\in K}  \right)a_{k}z\pb{p}_{k}(z)+
                          \sum_{j=1}^{n}zS_{j}(z)\pb{q}_{j}(z)\,.
  \end{align*}
  Observe that, due to \eqref{eq:p_relation},
  \begin{equation*}
    \sum_{k\in K^{\perp}}a_{k}z\pb{p}_{k}(z)=
    \sum_{k\in K^{\perp}}a_{k}\sum_{i=1}^{N}m_{ki}\pb{p}_{i}(z)\,,
  \end{equation*}
  where $m_{ki}$ are the entries of $M_{N}$, while
  \begin{equation*}
    \sum_{k\in K}a_{k}z\pb{p}_{k}(z)=
    \sum_{k\in K}a_{k}\left(\sum_{i=1}^{N}m_{ki}\pb{p}_{i}(z)+\pb{q}_{\gamma(k)}\right)\,,
  \end{equation*}
  where $\gamma(k)$ is such that $k$ is the $\gamma(k)$-th element
  of $K$. Thus,
   \begin{align*}
     \left(\sum_{k\in K^{\perp}} +\sum_{k\in K}\right)a_{k}z\pb{p}_{k}(z)
     &+\sum_{j=1}^{n}zS_{j}(z)\pb{q}_{j}(z)\\ &=\sum_{k=1}^{N}\tilde{a}_{k}\pb{p}_{k}(z) +
                                                \sum_{j=1}^{n}\left(a_{\gamma^{-1}(j)}+zS_{j}(z)\right)\pb{q}_{j}(z)
  \end{align*}
\end{proof}
For a fixed $M\in\mathfrak{M}$ and $N>n$, consider the corresponding
spectral function $\sigma_{N}^{\mathscr{T}}(t)$. Let
$\{\pb{e}_k\}_{k=1}^n$ be the canonical basis in $\complex^n$, i.e.,
\begin{equation}
  \label{eq:canonical-basis-CN}
  \pb{e}_1=\left(
    \begin{tiny}
      \begin{matrix}
        1\\
        0\\
        \vdots\\
        0
      \end{matrix}
    \end{tiny}
\right),\pb{e}_2=\left(
    \begin{tiny}
      \begin{matrix}
        0\\
        1\\
        \vdots\\
        0
      \end{matrix}
    \end{tiny}
\right),\dots,\pb{e}_n=\left(
    \begin{tiny}
      \begin{matrix}
        0\\
        \vdots\\
        0\\
        1
      \end{matrix}
    \end{tiny}
\right)\,.
\end{equation}
Define a family of vector
polynomials for $k\in\nats$ and $i=1,\dots,n$ as follows
\begin{equation}
    \label{eq:canonical-vector-polynomial}
    \pb{e}_{nk+i}(z):=z^k\pb{e}_i\,.
  \end{equation}
  These vector polynomials give an expression for the entries of the
  $(k+l)$-th matrix moment, $S_{k+l}(N)$, of the function
  $\sigma_{N}^{\mathscr{T}}$. Indeed,
\begin{align}
  \inner{\pb{e}_{i}}{S_{k+l}(N)\pb{e}_{j}}&=
     \inner{\pb{e}_{i}}{\int_{\reals}t^{k+l}d\sigma_{N}^{\mathscr T}(t)\pb{e}_{j}}\\
&=\inner{\pb{e}_{nk+i}(t)}{\pb{e}_{nl+j}(t)}_{L_2(\reals,\sigma_N^{\mathscr{T}})}\,.
\end{align}

\begin{theorem}
  \label{thm:moments-finite}
  Let $M\in\mathfrak{M}$. For any $m\in\nats$ there is $N\in\nats$
  such that
  \begin{equation*}
    \int_{\reals}t^k\,d\sigma_N^{\mathscr{T}}(t)=\int_{\reals}t^k\,d\sigma_{N^{\prime}}^{\mathscr{T}}(t)
  \end{equation*}
  for any $N^{\prime}\geq N$ and $k=0,1,\dots,2m-1,2m$, where
  $\sigma_{N}^{\mathscr T}$ and $\sigma_{N^{\prime}}^{\mathscr T}$ are
  the spectral functions of $M_{N}$ and $M_{N^{\prime}}$, respectively.
\end{theorem}
\begin{proof}
  Independently of the value of $N>n$, for $i,j=1,\dots,n$, one has
  \begin{equation}
    \label{eq:moments-vector-polynomials}
    \int_{\reals}t^{k+l}\,d\sigma_N^{\mathscr{T}}(t)=
    \inner{t^{k}\pb{e}_{i}}{t^{l}\pb{e}_{j}}_{L_2(\reals,\sigma_N^{\mathscr{T}})}=
    \inner{\mathscr T^{-1}t^{k}\pb{p}_{i}}{\mathscr T^{-1}t^{l}\pb{p}_{j}}_{L_2(\reals,\sigma_N^{\mathscr{T}})}
\end{equation}
since $\mathscr{T}\pb{e}_{i}=\pb{p}_{i}$ for $i\in G_{n}$ (which in
turn follows from \eqref{eq:psi_tau} and \eqref{eq:p_q_poly}).

Thus, the goal is to find a decomposition involving the vector polynomials
given in Definition~\ref{def:p_q_poly} for 
$z^{l}\pb{p}_{i}$ with $l\in\nats$ and $i\in G_{n}$. For the whole
infinite matrix $M$ and $i\in G_{j_{0}-1}$, one obtains from
\eqref{eq:q_relation_p} and \eqref{eq:p_relation} that
\begin{equation}
  \label{eq:invariant-decomposition2}
  z\pb{p}_{i}(z)=\sum_{k=1}^{k_{0}-1 }a_{k}^{i}\pb{p}_{k}(z)+
  \sum_{j=1}^{n-k_{0}+j_{0}}S_{j}^{i}(z)\pb{q}_{j}(z)\,.
\end{equation}
Also, directly from Definition~\ref{def:class}, one verifies that, for
any $m\in\nats\cup\{0\}$,
\begin{equation}
  \label{eq:tail-equations}
  z\pb{p}_{j_{0}+m}(z)=\sum_{j=1}^{k_{0}-1+m}m_{j_{0}+m,j}\pb{p}_{j}(z)+
  \pb{m}_{j_{0}+m,k_{0}+m}\pb{p}_{k_{0}+m}\,,
\end{equation}
where we have used the notation of Definition~\ref{def:class} for the
entries of $M$ and denote the edge entry in bold typeface.

In the case $j_{0}<n$, one recurs to
\eqref{eq:invariant-decomposition2} and \eqref{eq:tail-equations} to
find $z\pb{p}_{i}$ for $i\in G_{j_{0}-1}$ and
$i\in G_{n}\setminus G_{j_{0}-1}$, respectively (in
\eqref{eq:tail-equations} put $m=0,\dots,n-j_{0}$). If $j_{0}\ge n$,
then one only uses \eqref{eq:invariant-decomposition2} to determine
$z\pb{p}_{i}$ for $i\in G_{n}$.

As it was done in the proof of
Theorem~\ref{thm:decomposition-marchenko} an induction argument, which
involves \eqref{eq:invariant-decomposition2} and 
\eqref{eq:tail-equations}, yields
\begin{equation}
  \label{eq:invariant-decomposition}
  z^{l}\pb{p}_{i}(z)=\sum_{k=1}^{i+l(k_{0}-j_{0})}a_{k}^{i}(l)\pb{p}_{k}(z)+
  \sum_{j=1}^{n-k_{0}+j_{0}}S_{j}^{i}(l,z)\pb{q}_{j}(z)\,,\quad
  i\in G_{n}
\end{equation}
for any $l\in\nats$, where we have indicated the dependence on $l$ of
the coefficients. If one considers the finite submatrix $M_{N}$, then
\eqref{eq:invariant-decomposition} remains invariant as long as
$N\ge n+l(k_{0}-j_{0})$. This invariance is a consequence of the fact
that all equations in \eqref{eq:tail-equations} correspond to rows of
the matrix $M$ on the tail (see (2) of
Definition~\ref{def:class}).

Now, it follows from \eqref{eq:moments-vector-polynomials},
\eqref{eq:invariant-decomposition2},
\eqref{eq:invariant-decomposition}, and
Propositions~\ref{prop:ortonormal-p-L2-finite} and
\ref{prop:zeros-finite} that the entries of the moment matrix
$S_{s}(N)$, with $s\in G_{2l}$ and $N\ge n+l(k_{0}-j_{0})$, are given
by $\mathscr T$ and the coefficients $a_{k}^{i}(l)$, where
$k\in G_{n+l(k_{0}-j_{0})}$ and $i\in G_{n}$.
\end{proof}

Having Theorem~\ref{thm:moments-finite} at our disposal, the following
assertion is a direct consequence of \cite[Lems.\,3.1,\,3.2 and
Prop.\,3.1]{MR3543793}. The proof is omitted here since the reasoning
repeats the one used in \cite{MR3543793}, which in turn relies on
generalizations of Helly's theorems (see
\cite[Thms.\,4.3,\,4.4]{MR0322542}) and the argumentation found in
\cite[Sec.\,2.1]{MR0184042}. Note, however, that
Theorem~\ref{thm:moments-finite} is more complex than its counterpart
in the more simple setting treated in \cite{MR3543793}.

\begin{theorem}
  \label{thm:helly-theorems}
  Let $M\in\mathfrak{M}$ and $l\in\nats$ and consider the spectral
  functions $\sigma_{N}^{\mathscr T}$ of $M_{N}$ with
  $N\ge n+l(k_{0}-j_{0})$ for a certain boundary matrix $\mathscr T$
  (see Section~\ref{sec:spectr-meas-subm}), then there exists a
  subsequence $\{\sigma_{N_{j}}^{\mathscr T}\}_{j=1}^{\infty}$
  converging pointwise to a matrix-valued function $\sigma$ such that
\begin{equation}
  \label{eq:infinite-moments}
  \int_{\reals}t^{k}d\sigma_{N_{j}}^{\mathscr T}(t)=
  \int_{\reals}t^{k}d\sigma(t)
\end{equation}
for any nonnegative integer $k\le 2l$.
\end{theorem}

As a result of Theorem~\ref{thm:moments-finite}, the matrix moments of a
spectral function do not depend on $N\gg 1$. Moreover, due to
Theorem~\ref{thm:helly-theorems}, one can consider a sequence of moments
$\{S_k\}_{k=0}^\infty$ given by
\begin{equation*}
  S_{k}:=\int_{\reals}t^{k}d\sigma_{N}^{\mathscr T}(t)
\end{equation*}
for any $N\ge n+l(k_{0}-j_{0})$ and $k\le 2l$ and, consequently, pose the corresponding
moment problem.

\begin{corollary}
  The spectral function $\sigma$ to which the subsequence
  $\{\sigma_{N_{j}}^{\mathscr{T}}\}_{j=1}^\infty$ converges according to
  Theorem~\ref{thm:helly-theorems} is a solution with an infinite
  number of growing points to a certain
  matrix moment problem given by the sequence $\{S_k\}_{k=0}^\infty$.
\end{corollary}

\begin{definition}
  \label{def:spectral-measure-gen-case}
  A nondecreasing $n\times n$ matrix-valued function $\sigma$ with an infinite
  number of growing points and finite moments is said to belong to the
  class $\mathfrak S_{n}$ when $\int_\reals d\sigma$ is invertible. A
  function $\sigma\in\mathfrak S_{n}$ is
  called a spectral function of a matrix $M$ in $\mathfrak{M}$
  (actually we mean the spectral function of $\mathcal{M}$; see
  Section~\ref{sec:preliminaries}) when there exist a boundary matrix
  $\mathscr{T}$ such that $\{\pb{p}_{k}\}_{k=1}^\infty$ is an
  orthonormal sequence in $L_2(\reals,\sigma)$.
\end{definition}
As in the case when $M\in\widetilde{\mathfrak{M}}$, studied in
\cite[Sec.\,3]{MR3543793}, it is verified that
the spectral function $\sigma$ of $M\in\mathfrak{M}$
 has an infinite number of growth
 points. Moreover, the vector polynomials are densely contained in
 $L_2(\reals,\sigma)$ whenever the orthonormal
system $\{\pb{p}_k\}_{k=1}^\infty$ turns out to be complete.

By Definition~\ref{def:spectral-measure-gen-case}, one
constructs an isometry $\mathcal U$ between $l_{2}(\nats)$ and
the closure of the polynomials in
$L_2(\reals,\sigma)$ by associating the orthonormal basis
$\{\delta_k\}_{k=1}^\infty$ with the orthonormal system
$\{\pb{p}_k\}_{k=1}^\infty$, i.\,e., $\mathcal U\delta_k=\pb{p}_k$ for all
$k\in\nats$. Under this map, the operator $\mathcal M$ is
transformed into some restriction of the operator of multiplication by
the independent variable. Indeed, if
$v=\sum_{k=1}^\infty v_k\delta_k$ is an element of the
domain of $\mathcal M$, then $\pb{v}=\sum_{k=1}^\infty v_k\pb{p}_k$ is in the
domain of the operator of multiplication by the independent variable
and
\begin{equation*}
  \mathcal U\mathcal{M}\mathcal U^{-1}\pb{v}(t)=t\pb{v}(t)\,.
\end{equation*}
\begin{lemma}
  \label{lem:has-a-spectral-function}
  Any $M\in\mathfrak{M}$ has at least one spectral
  function (as given in
  Definition~\ref{def:spectral-measure-gen-case}).
\end{lemma}
\begin{proof}
  It is a consequence of 
  Proposition~\ref{prop:ortonormal-p-L2-finite} and
  Theorem~\ref{thm:helly-theorems} that the vector
  polynomials $\{\boldsymbol{p}_k(z)\}_{k=1}^{\infty}$, given in \eqref{eq:p_q_poly}, satisfy
\begin{equation*}
  \inner{\boldsymbol{p}_j}{\boldsymbol{p}_k}
_{L_2(\mathbb{R},\sigma)}
=\delta_{jk}
\end{equation*}
for $j,k\in\nats$, where $\sigma$ is the function given
in Theorem~\ref{thm:helly-theorems}.
\end{proof}
\begin{lemma}
  \label{lem:spectral-function-zeroes}
  If $\sigma$ is a spectral function of  $M\in\mathfrak{M}$, then each
  vector polynomial $\pb{q}_{j}$, $j\in G_{n-k_{0}+j_{0}}$, has zero norm.
\end{lemma}
\begin{proof}
 If one fixes
$j\in G_{n-k_{0}+j_{0}}$ and considers $N\gg 1$, then
  \begin{equation*}
    0=\norm{\pb{q}_j}_{L_2(\reals,\sigma_N^{\mathscr{T}})}^2=
    \int_\reals\inner{\pb{q}_j}{d\sigma_N^{\mathscr{T}}\pb{q}_j}
  \end{equation*}
as a result of Proposition~\ref{prop:zeros-finite}.
Now, by Theorem~\ref{thm:helly-theorems} there is a
subsequence  $\{\sigma_{N_i}^{\mathscr{T}}\}_{i=1}^{\infty}$ such
that
\begin{equation*}
 0= \int_\reals\inner{\pb{q}_j}{d\sigma_{N_i}^{\mathscr{T}}\pb{q}_j}=
\int_\reals\inner{\pb{q}_j}{d\sigma\pb{q}_j}=
\norm{\pb{q}_j}_{L_2(\reals,\sigma)}^2\,.
\end{equation*}
\end{proof}

\section{Reconstruction of the matrix}
\label{sec:Reconstruction}

The goal of this section is to reconstruct a matrix from an arbitrary
function $\sigma\in\mathfrak S_{n}$ (see
Definition~\ref{def:spectral-measure-gen-case}). This has already
being done in \cite[Sec.\,5]{MR3543793}. Indeed, as was shown in the
previous section, the matrix-valued function obtained in
Theorem~\ref{thm:helly-theorems} is a solution to a matrix moment
problem and has the same properties of the functions constructed in
\cite[Sec.\,4]{MR3543793}.

The following
statement is a rephrasing of \cite[Lem.\,5.1, Prop.\,5.1]{MR3543793}.
\begin{theorem}
  \label{thm:gram-schmidt}
  Given a function $\sigma\in\mathfrak S_{n}$ there is an infinite
  sequence of orthonormal polynomials $\{\tb{p}_k: k\in\nats\}$
  in $L_{2}(\reals,\sigma)$ and a finite sequence $\{\tb{q}_j:j\in G_{n_{0}}\}$, $n_{0}<n$ such
  that  $\norm{\tb{q}_j}=0$ for any $j\in G_{n_{0}}$. Moreover for any
  $n$-dimensional vector polynomial $\pb{r}$ there is $l\in\nats$ such
  that
  \begin{equation*}
    \pb{r}(z)=\sum_{k=1}^{l}a_{k}\tb{p}_{k}(z)+\sum_{j=1}^{n_{0}}S_{j}(z)\tb{q}_{j}(z)\,,
  \end{equation*}
  where each $a_{k}$ is a complex number and each $S_{j}$ is a scalar
  polynomial.
\end{theorem}
The proof of this assertion if based on the Gram-Schmidt
procedure taking into account the properties of a function in
$\mathfrak{S}$. The algorithm is illustrated in the flow chart of
Fig.\ref{fig:flow}.

\tikzstyle{decision} = [diamond, draw, text width=7.6em, text badly
centered, node distance=4cm, inner sep=0pt, aspect=2, rounded corners]
\tikzstyle{decision1} = [diamond, draw, text width=4.5em, text badly
centered, node distance=4cm, inner sep=0pt, aspect=2, rounded corners]
\tikzstyle{block} = [rectangle, draw, text width=5em, text centered,
rounded corners, minimum height=3em] \tikzstyle{block1} = [rectangle,
draw, text width=8em, text centered, rounded corners, minimum
height=3em] \tikzstyle{line} = [draw, -latex'] \tikzstyle{cloud} =
[draw, ellipse, node distance=3cm, minimum height=2em, text width=3em,
text centered]
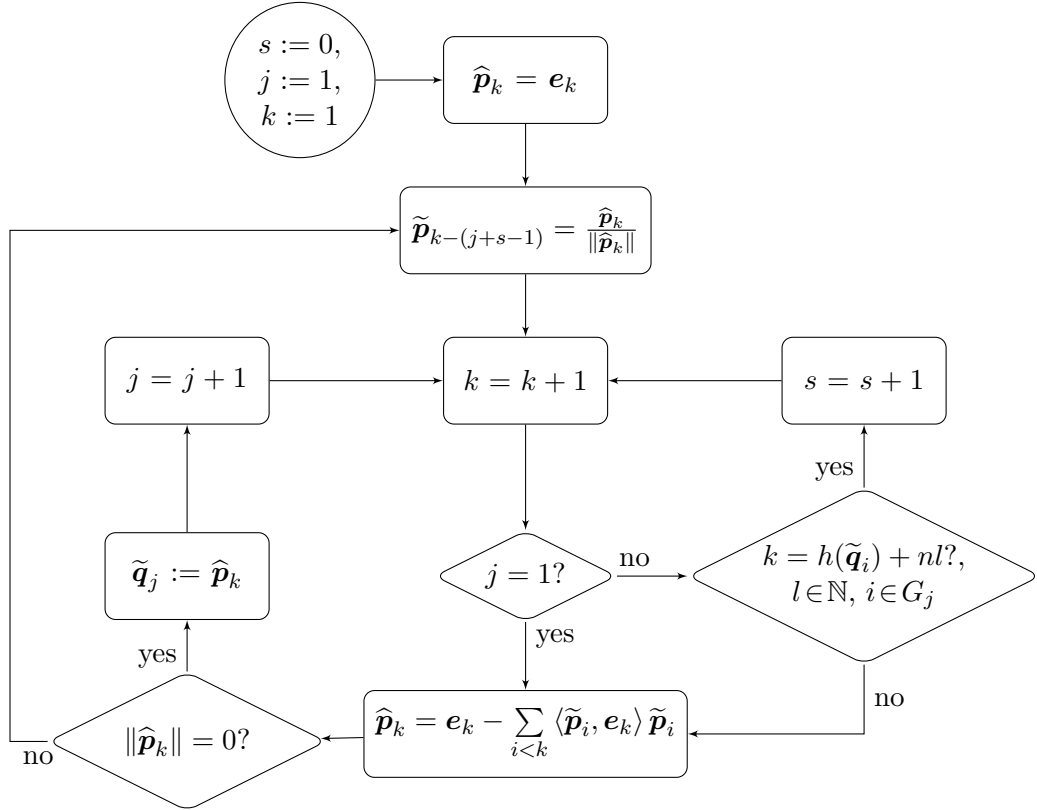
\begin{figure}[h]
  \label{fig:flow}
\begin{small}
  \begin{center}
    \begin{tikzpicture}[node distance = 3cm, auto]
      \node [block] (init) {$\widehat{\pb{p}}_k=\pb{e}_k$}; \node
      [cloud, left of=init] (initial-condition) {$s:=0$, $j:=1$,
        $k:=1$}; \node [block1, below of=init, node distance=2cm]
      (normalized)
      {$\tb{p}_{k-(j+s-1)}=\frac{\widehat{\pb{p}}_k}{\norm{\widehat{\pb{p}}_k}}$};
      \node [block, below of=normalized, node distance=2cm]
      (k-counter) {$k=k+1$}; \node [block, right of=k-counter, node
      distance=4.5cm] (s-counter) {$s=s+1$}; \node [block, left
      of=k-counter, node
      distance=4.5cm] (j-counter) {$j=j+1$}; \node [decision1, below
      of=k-counter, node distance=2.6cm] (j-1) {$j=1$?}; \node
      [decision, right of=j-1, node distance=4.5cm] (k-mod-n) {
        $k=h(\tb{q}_{i})+nl$?, $l\!\in\!\nats,\,i\!\in\!G_j$};
      \node [block, below of=j-counter, node distance=26mm] (q-poly)
      {$\tb{q}_j:=\widehat{\pb{p}}_k$}; \node [block1, below of=j-1,
      node distance=2.1cm, text width=10.5em] (G-S)
      {$\widehat{\pb{p}}_k=\pb{e}_k-\sum\limits_{i<k}
        \inner{\tb{p}_i}{\pb{e}_{k}}\tb{p}_i$}; \node [decision, below
      of=q-poly, node distance=2.2cm] (zero-norm) {
        $\norm{\widehat{\pb{p}}_k}=0$?};
      \path [line] (initial-condition) -- (init); \path [line] (init)
      -- (normalized); \path [line] (normalized) -- (k-counter); \path
      [line] (k-counter) -- (j-1); \path [line] (j-1) -- node [near
      start] {yes} (G-S); \path [line] (G-S) -- (zero-norm); \path
      [line] (zero-norm.west) -- node [near start] {no} ++(-0.5cm,0) |-
      (normalized); \path [line] (zero-norm) -- node [near start]
      {yes} (q-poly); \path [line] (q-poly) -- (j-counter); \path
      [line] (s-counter)--(k-counter); \path [line]
      (j-counter)--(k-counter); \path [line] (j-1)-- node [near start]
      {no} (k-mod-n); \path [line] (k-mod-n)-- node [near start] {yes}
      (s-counter); \path [line] (k-mod-n)|- node [near start] {no}
      (G-S);
    \end{tikzpicture}
  \end{center}
\end{small}
\caption{Orthonormalization algorithm}
\end{figure}
The properties of the vector polynomials obtained by the algorithm above
described yield the following assertion \cite[Thms.\,6.4,6.5]{MR3543793}) 
\begin{theorem}
  \label{thm:reconstruction-spectral-function}
  If $\sigma\in\mathfrak S_{n}$, then there is a matrix $M\in\mathfrak
  M$ and a boundary matrix $\mathscr T$ such that $\sigma$ is the
  corresponding spectral measure of $M$.
\end{theorem}
The results of the previous section and the last two theorems give the
following statement.
\begin{theorem}
  \label{thm:main}
  If $M\in\mathfrak M$ and $\sigma$ is a spectral function given in
  Definition~\ref{def:spectral-measure-gen-case}, then there is
  $\widetilde M\in\widetilde{\mathfrak M}$ and a boundary matrix
  $\mathscr T$ such that $\sigma$ is also a spectral function of
  $\widetilde M$.
\end{theorem}

This proves that any operator associated with a matrix
in $\mathfrak M$ is unitary equivalent to a operator associated with a
matrix in $\widetilde{\mathfrak{M}}$. This equivalence gives a new
approach to the problem of linear multidimensional interpolation and
also shed light on mechanical systems, in particular on the
equivalence between systems of interacting particles.

\subsection*{Acknowledgments}
S.P. was supported by CONAHCYT Apoyos Complementarios para Estancia
Sabática Vinculadas a la Consolidación de Grupos de Investigación \textnumero\,852558. L.O.S. has been partially supported by
CONACyT Ciencia de Frontera 2019 \textnumero\,304005.
\def\cprime{$'$} \def\lfhook#1{\setbox0=\hbox{#1}{\ooalign{\hidewidth
  \lower1.5ex\hbox{'}\hidewidth\crcr\unhbox0}}} \def\cprime{$'$}
  \def\cprime{$'$} \def\cprime{$'$} \def\cprime{$'$} \def\cprime{$'$}
  \def\cprime{$'$} \def\cprime{$'$}

\end{document}